\documentclass[11pt]{article}
\usepackage{mathpazo}
\usepackage[T1]{fontenc}
\usepackage[letterpaper, margin=1in]{geometry}
\usepackage{amsmath,amssymb,amsfonts,amsthm}
\usepackage{mathtools}
\usepackage{bm}
\usepackage{bbm}
\usepackage[dvipsnames]{xcolor}
\usepackage{xspace}
\usepackage{microtype}
\usepackage{csquotes}
\usepackage{graphicx}
\usepackage{subcaption}
\usepackage{adjustbox}
\usepackage{float}
\usepackage{enumitem}
\setlist{nosep,topsep=0pt,leftmargin=*}
\usepackage{url}
\usepackage{alphabeta}

\usepackage{todonotes}
\usepackage{authblk}

\usepackage{algorithm}
\usepackage{algpseudocode}

\definecolor{myRed}{rgb}{0.82,0.13,0.56}
\definecolor{myBlue}{RGB}{13,55,174}
\usepackage{hyperref}
\hypersetup{
  colorlinks=true,
  citecolor=myRed,
  linkcolor=myBlue,
  urlcolor=PineGreen
}
\usepackage[noabbrev,capitalise]{cleveref}

\usepackage{thmtools}
\usepackage{thm-restate}

\usepackage{eqparbox}

\newcommand{\ceil}[1]{\left\lceil #1 \right\rceil}

\newcommand{\cost}{\mathrm{C}}

\newcommand{\OPT}{\textsc{Opt}}

\newcommand{\range}{\mathrm{range}}

\theoremstyle{plain}
\newtheorem{theorem}{Theorem}[section]
\newtheorem{proposition}[theorem]{Proposition}
\newtheorem{lemma}[theorem]{Lemma}

\newtheorem{corollary}[theorem]{Corollary}
\theoremstyle{definition}
\newtheorem{definition}[theorem]{Definition}

\newtheorem{remark}[theorem]{Remark}
\theoremstyle{remark}

\usepackage[numbers]{natbib}

\author[1,2]{Charalampos Platanos}
\author[1,2]{Thanos Tolias}

\affil[1]{National Technical University of Athens, Greece}
\affil[2]{Archimedes RU, Athena RC, Greece}

\affil[ ]{\texttt{{
\href{mailto:harrisplat@gmail.com}{harrisplat@gmail.com} \quad
\href{mailto:thanostolias@mail.ntua.gr}{thanostolias@mail.ntua.gr}}}}

\title{Online Sorting with Our Eyes Wide Shut}
\date{}

\begin{document}

\maketitle

\renewcommand\thefootnote{}\footnotetext{This work has been partially supported by project MIS 5154714 of the National Recovery and Resilience Plan Greece 2.0 funded by the European Union under the NextGenerationEU Program.}

\begin{abstract}
In \emph{Online Sorting}, we are given an array $A$ of $n$ initially empty cells. At each time step $t\in[n]$, an element $x_t\in[0,1]$ arrives and must be placed irrevocably into an empty cell, without knowledge of future arrivals. The objective is to minimize the sum of absolute differences between elements assigned to adjacent cells. The problem has been studied under both adversarial and stochastic input models. For adversarial sequences, Aamand, Abrahamsen, Beretta, and Kleist (SODA'23) gave a tight $O(\sqrt n)$-competitive algorithm, fully resolving the worst-case setting. For stochastic sequences, in which the elements are drawn i.i.d.\ from $U[0,1]$, the picture is settled: Hu (SODA'26) gave an $\log n\cdot 2^{O(\log^* n)}$-competitive algorithm in expectation and proved an $\Omega(\log n)$ lower bound, while Kalavas, Platanos, and Tolias (STACS'26) gave an $O(\log^2 n)$-competitive algorithm with high probability. Very recently, Hermansen (ESA'26) closed the remaining gap by designing an $O(\log n)$-competitive algorithm in expectation.

In this work, we study \emph{Random-Order Online Sorting}, a model interpolating between the adversarial and stochastic settings, that was posed as a challenging open question by Hermansen (ESA'26). Here, the input is a multiset chosen adversarially, but its elements arrive in uniformly random order. We take a different point of view by solving the problem in rank space, and prove an $O(\log^2 n)$-competitive algorithm with high probability, matching the state-of-the-art high probability guarantee for the stochastic setting in this more general model. A central challenge is that the algorithm must learn a useful rank partition from samples obtained online, while already committing arrivals irrevocably and without incurring excessive cost during this learning process.

We also study a multidimensional generalization, which we call \emph{Random-Order Online TSP}, and obtain an $O(\log^3 n)$-competitive algorithm with high probability. Our construction uses random prefixes to build near-optimal reference curves, which reduce each phase of the multidimensional problem to an instance of one-dimensional random-order online sorting.

\end{abstract}

\newpage
\section{Introduction}

In the \emph{Online Sorting} problem, we are given an array \(A\) consisting of \(n\) initially empty cells. At each time step \(t\in [n]\), an element \(x_t\in [0,1]\) arrives and must be placed irrevocably into one of the currently empty cells of the array. After all elements have been placed, the cost of the resulting array is the sum of the absolute differences between values in adjacent cells of the array. The objective is to minimize this cost.

The study of Online Sorting was initiated by~\citet*{abrSODA}, who introduced the problem as a tool for proving lower bounds in online packing problems. They studied the adversarial variant of the problem, in which the input sequence is chosen and revealed adversarially, and gave an \(O(\sqrt n)\)-competitive algorithm together with a matching lower bound for deterministic algorithms. This established a sharp picture for the one-dimensional problem in the worst case and showed that, despite the simplicity of the offline optimum, the online problem is inherently nontrivial.

Subsequent work broadened this picture by considering more structured input models.~\citet*{AbrESA} introduced the \emph{stochastic} variant, in which the input elements are drawn i.i.d.\ from the uniform distribution on \([0,1]\), and showed that this additional randomness can be exploited to obtain significantly better guarantees; in particular, they gave an \(\widetilde{O}(n^{1/4})\)-competitive algorithm. Hu~\cite{hu2025} obtained a \(\log n\cdot 2^{O(\log^* n)}\)-competitive algorithm together with an \(\Omega(\log n)\) lower bound, while~\citet*{KPT} gave an \(O(\log^2 n)\)-competitive algorithm with high probability. More recently, \citet{hermansen26} \footnote{We cite the author’s Master’s thesis, which contains the results of the paper accepted at ESA 2026, since the conference version does not appear to be publicly available; the paper’s acceptance is indicated on the conference website’s list of accepted papers.} provided an $O(\log n)$-competitive algorithm in expectation, closing the remaining gap. These results suggest that the difficulty of the problem is highly sensitive to the way the input is generated.

In this work, we consider a natural intermediate model, which we call
\emph{Random-Order Online Sorting}. Here, an adversary chooses the multiset
\(X\) of input values, but the elements are revealed in uniformly random order.
Random-order models are a standard way to interpolate between fully adversarial
and stochastic input assumptions: they preserve the worst-case structure of the
instance, while ruling out adversarial control over the order in which this
structure is exposed. In our setting, this distinction is particularly natural.
The algorithm receives no prior distributional information, but random prefixes
nevertheless provide increasingly representative information about the rank
structure of the input. Thus, the random-order model asks whether the strong
guarantees known in stochastic settings rely on independence and knowledge of
the input distribution, or whether random permutation alone already suffices.

Our main result shows that the latter is the case: we obtain an
\(O(\log^2 n)\)-competitive algorithm with high probability for
Random-Order Online Sorting, thus answering the main open questions of \cite{KPT, hermansen26}.

A related line of work considers multidimensional extensions of the problem.~\cite{AbrESA} introduced \emph{Online TSP}, in which the input consists of points in \([0,1]^d\) and the goal is to minimize the length of the resulting path through the assigned points. They gave an \(\widetilde{O}(\sqrt n)\)-competitive algorithm for this problem in fixed dimension. Shortly after, Bertram~\cite{bertESA} generalized this picture by giving an \(O(\sqrt n)\)-competitive algorithm for Online Metric TSP in arbitrary metric spaces, thereby in particular removing the dependence on the dimension in the Euclidean setting.

The authors in ~\cite{KPT} studied the stochastic variant of Online TSP, in which the points are drawn i.i.d.\ from the uniform distribution on the \(d\)-dimensional unit cube, and obtained polylogarithmic guarantees with high probability. Recently, perhaps surprisingly, \cite{hermansen26} designed $O(1)$ competitive algorithms for $d\geq 2$, suggesting that the TSP variant is easier than the single-dimensional. Their works left open whether such guarantees can be obtained more generally for arbitrary known input distributions. We answer this question in the affirmative. In this paper, we study the random-order model, which we call \emph{Random-Order Online TSP}, and obtain an \(O(\log^3 n)\)-competitive algorithm with high probability. Since our guarantee holds for every adversarially chosen multiset and uses only the random arrival order, it immediately implies polylogarithmic guarantees for inputs drawn i.i.d.\ from any, even unknown, distribution \(D\). Note that one cannot hope for a $o(\log n)$-competitive algorithm in our Random-Order Online TSP (or even the known but adversarially selected distribution model), since the single-dimensional lower bounds are inherited. 

\subsection{Technical Overview and Contributions}

We begin with our result for Random-Order Online Sorting.

\begin{theorem}
There exists an \(O(\log^2 n)\)-competitive algorithm with high probability for Random-Order Online Sorting.
\end{theorem}

To understand the algorithm, it is helpful to first consider an idealized setting. Suppose that the array could be partitioned into many small contiguous blocks of empty cells, which we call buckets, and that the value space could simultaneously be partitioned into intervals so that each bucket receives exactly the points belonging to its corresponding interval. If the buckets appear in the same left-to-right order as the intervals, then we may handle each bucket independently using the adversarial online placement routine of~\cite{abrSODA}. Since every bucket only sees points from a short interval, the in-bucket cost is small, and since neighboring buckets correspond to neighboring intervals, the total boundary cost is also small. In other words, once such a partition is available, the problem essentially decomposes into many local one-dimensional subproblems.

This ideal picture leaves us with two difficulties. First, even if we are given
a prefix of the input, how should we use it to construct an informative
partition of the instance? Second, how can we obtain such a prefix online
without already incurring a large cost before the partition has been learned?
Existing stochastic algorithms do not directly answer these questions. In the
i.i.d. setting, the algorithm may exploit distributional information, or the fact
that the geometry of future arrivals is governed by the same distribution as the
past. In the random-order model there is no underlying distribution known to the
algorithm, and the observed prefix and the future suffix are negatively
correlated samples from an adversarially chosen multiset.

Our answer is to view the problem through the lens of rank space. Rather than
trying to partition the geometric interval \([0,1]\) directly, we use the
observed prefix to build an empirical quantile partition. That is, we sort the prefix and cut it into blocks of equal rank size, thereby obtaining intervals that each contain roughly the same number of sampled points. This point of view is robust to the actual geometry of the instance: what matters is not where the values lie numerically, but how future arrivals are distributed relative to the order statistics of the prefix. More importantly, rank space turns the load of a bucket into a probabilistic question that the random-order model lets us analyze cleanly. The key probabilistic insight is easiest to see in rank space. Fix a prefix of the input and consider some later collection of arrivals. If we condition on the set of points appearing in the prefix and in this later portion of the stream, then their relative order can be represented by a uniformly random binary string, indicating which ranks belong to the prefix and which to the later arrivals. Consequently, the number of future points falling between two sampled ranks is governed by the number of zeros appearing between the corresponding ones. This allows us to show that intervals defined by empirical quantiles of the observed prefix receive nearly the expected number of future arrivals. We then assign a contiguous bucket of array cells to each such interval and place each arriving element into the earliest eligible bucket with remaining capacity. The small imbalances that remain are absorbed by the subsequent evolution of the algorithm.

The second question is more subtle, because in the random-order model the partition is not known in advance and must be learned from the same stream on which the algorithm is already acting. Our main algorithmic contribution is an \emph{ascent period} that gathers this information online. During the ascent, the algorithm repeatedly pauses at geometrically chosen update times, looks at the prefix revealed so far, and refines its empirical quantile partition accordingly. Each refinement creates a new phase with a finer collection of intervals and buckets. Intuitively, the ascent allows the algorithm to bootstrap itself: early coarse phases absorb the cost of learning, while later phases exploit the increasingly accurate rank-space information extracted from larger prefixes.

The ascent alone is not enough, because the empirical intervals obtained from a prefix are only approximately balanced, and so each new phase may inherit a small load discrepancy. To absorb these discrepancies, the algorithm continues with a middle phase at the finest scale and then a sequence of descent phases that gradually coarsen the partition again. These later phases provide the slack needed to stabilize the process while keeping all bucket sizes polylogarithmic.

The analysis is organized around a simple invariant: at every update, the buckets of the immediately preceding phase are almost full, while the buckets of the phase before that are completely full. Once this invariant is established, the remainder of the argument is deterministic. It guarantees that every new phase can be constructed, that every arriving element always finds an eligible bucket with available space, and that the total cost can be bounded by summing the in-bucket and boundary contributions over all phases. Since every bucket has size $O(\log^2 n)$, this yields the desired $O(\log^2 n)$ competitive ratio\footnote{By adjusting the constants and refining the analysis so that buckets have size $O(\log n)$, one can improve the competitive ratio to $O(\log^{3/2} n)$. We use the present formulation for clarity (see also Remark 13 in \cite{KPT} and Section 2.3.3 in \cite{hermansen26}).}.

We next turn to the multidimensional extension.

\begin{theorem}
There exists an \(O(\log^3 n)\)-competitive algorithm with high probability for Random-Order Online TSP.
\end{theorem}

The multidimensional case presents a different obstruction. In the uniform
stochastic setting, one can exploit the ambient geometry of the cube: random
prefixes are geometrically representative of the distribution from which future
points are drawn. This kind of argument is no longer available for an arbitrary
adversarial multiset. The input may concentrate on an unknown lower-dimensional
or highly irregular structure, and the algorithm has no prior distributional
description from which to build a reference ordering. Moving to higher dimensions also removes the canonical order available on the
line. Thus, before we can reuse our one-dimensional machinery, we first need to
extract from the revealed prefix a one-dimensional structure that is both
algorithmically meaningful and cheap relative to the offline optimum.

Our approach is to build such a structure dynamically from the input itself. We process the instance in doubling phases. At the beginning of each phase, we look at the prefix revealed so far and compute a shortest Hamiltonian path through those points. This path serves as a \emph{prefix curve}, giving a one-dimensional reference structure adapted to the geometry of the observed prefix. The key idea is then to project the points arriving in the current phase onto this curve and to use their positions along the curve as one-dimensional keys. In this way, each phase is converted into an instance of Random-Order Online Sorting.

The main issue is to show that this projection does not distort the objective too much. The first part is clean: since the prefix curve is a shortest Hamiltonian path through only a subset of the input, its length is at most \(\OPT\) for the full instance. The more delicate part is to bound the total distance of the current phase points from the prefix curve. Here the random-order assumption is crucial. Condition on the set of points revealed by the end of the current phase. Then the prefix is simply a uniformly random half of that set, and the current phase consists of the complementary half. If we examine an optimal tour of this combined set, the points of the current phase appear in random runs between prefix points. With high probability, no such run is too long. Consequently, every point in the current phase lies close to some prefix point along the tour, and summing over all runs shows that the total projection cost is only \(O(\log n)\cdot \OPT\).

Once this bound is established, the rest of the reduction is straightforward. Applying the one-dimensional theorem to the projected keys yields an \(O(\log^2 n)\) factor within each doubling phase, while summing over the \(O(\log n)\) phases incurs one additional logarithmic loss. This gives an \(O(\log^3 n)\)-competitive algorithm with high probability for Random-Order Online TSP.

\subsection{Further Related Work}

\subparagraph{Hashing.} The connection between stochastic Online Sorting and hashing has been apparent since the problem’s introduction. Indeed, a natural first attempt at solving stochastic Online Sorting is to apply \emph{linear probing}~\cite{Knuth}. In this approach, an arriving element is placed at the position corresponding to its value percentile; for the uniform distribution on $[0,1]$, an element $x$ is mapped to position $\lceil x\cdot n\rceil$. If each position were filled by exactly one element mapped to it, this strategy would be optimal.

However, with high probability, multiple elements are mapped to the same position. Linear probing resolves such collisions by placing the element in the first available cell to the left of its target position. As collisions accumulate, the connection becomes clearer; resolving collisions by displacing elements to neighboring positions increases the respective objectives in a comparable manner (namely, the probe complexity in hashing and the adjacency cost in Online Sorting).

While the literature on hashing is vast, two schemes are particularly relevant in this context: \emph{Filter Hashing}~\cite{fotakis2003} and \emph{Transactional Multi-Writer Cuckoo Hashing}~\cite{kuszmaul2016}. Both employ a multi-layered architecture in which successive layers absorb and correct load imbalances created by earlier ones, a mechanism that closely parallels the phased and adaptive structure underlying our algorithms.

\subparagraph{Online Sorting with Larger Arrays.}
The formation of large contiguous clusters of occupied cells is a central obstacle for Online Sorting algorithms that follow linear-probing–style placement rules. A natural variant of the problem investigates how disruptive such clusters remain when additional array capacity is available. In this setting, the array length $m$ exceeds the number of arrivals $n$, i.e., $m>n$, and the goal is to understand how the competitive ratio improves as a function of the slack $m/n$. This variant was introduced by ~\citet*{abrSODA}, who designed a deterministic $2^{\sqrt{\log n}\sqrt{\log\log n + \log(1/\varepsilon)}}$-competitive algorithm for $m=(1+\varepsilon)n$. They also proved a lower bound, showing that every deterministic algorithm with $m=\gamma n$ is at least $\frac{1}{\gamma}\cdot \Omega(\log n/\log\log n)$-competitive. Subsequently, in independent and concurrent work,~\citet*{azar2025} and ~\citet*{nirjhor2025} improved the upper bound, nearly resolving this variant. More recently,\citet*{azar26} extended this line of work to the multidimensional setting of Online Metric TSP. They showed that, with an array of size $m=(1+\varepsilon)n$, there is a deterministic $O(\log^3 n/\varepsilon)$-competitive algorithm. For the stochastic case, \citet{hermansen26} proved that, with an array of size $m=\ceil{(1+\varepsilon)n}$, there is an $O(1+\log(1/\varepsilon))$-competitive algorithm.

\section{Preliminaries}

\begin{definition}
For a nonempty multiset \(Y\subseteq[0,1]\), define $\range(Y)=\max Y-\min Y$.
\end{definition}

\begin{definition}
In \emph{Random-Order Online Sorting}, an adversary chooses a multiset $X=\{x_1,\ldots,x_n\}\subseteq[0,1]$, and the elements of \(X\) are revealed in uniformly random order. We are given an array \(A\) of \(n\) initially empty cells, and at each time step \(t\in[n]\), the arriving element must be placed irrevocably into an empty cell of \(A\). After all elements have been placed, the cost of the resulting array is
\[
    C(A)=\sum_{i=2}^n |A[i]-A[i-1]|.
\] 
The offline optimum for an instance \(X\) is $\OPT(X)=\range(X)$.
\end{definition}

\begin{definition}
In \emph{Random-Order Online TSP}, an adversary chooses a multiset $X=\{x_1,\ldots,x_n\}\subseteq [0,1]^d$, and the points of \(X\) are revealed in uniformly random order. As before, we are given an array \(A\) of \(n\) initially empty cells, and each arriving point must be assigned irrevocably to an empty cell. The cost of the resulting array
is
\[
    C(A)=\sum_{i=2}^n \|A[i]-A[i-1]\|,
\]
and we denote by \(\OPT(X)\) the length of a shortest Hamiltonian path through the points of \(X\).  
\end{definition}

Let \(\pi\) denote the uniformly random permutation determining the arrival order of \(X\), and let \(R\) denote the internal randomness of the algorithm. An online algorithm is \(c\)-competitive in expectation if, for every adversarially chosen multiset \(X\) of size \(n\), $
    \mathbb{E}_{\pi, R}[{\cost(A)}] \le c\OPT(X),$
where the expectation is taken over both the random arrival order and the
internal randomness of the algorithm. We say that an online algorithm is \(c\)-competitive with high probability if, for every adversarially chosen multiset \(X\) of size \(n\),
\[
    \Pr_{\pi,R}\!\left[
        \cost(A)\le c\OPT(X)
    \right]
    \ge 1-\frac{1}{n},
\]
where the probability is again taken over both the uniformly random arrival order and the internal randomness of the algorithm.

\section{Random-Order Online Sorting}\label{sec: Sorting}

To motivate our approach, we first recall a result of~\citet*{abrSODA}, which will serve as our low-level placement subroutine.

\begin{theorem}[\cite{abrSODA}]\label{thm:adversarial}
There is an absolute constant \(\xi>0\) and a deterministic online routine
\textsc{AdversarialPlacement} with the following guarantee. For every sequence
\(x_1,x_2,\ldots,x_k\), with \(x_i\in[0,1]\) for each \(i\), if all \(k\)
values lie in an interval \(I\subseteq[0,1]\), then
\textsc{AdversarialPlacement} can place them into any prescribed contiguous
block of \(k\) cells with total cost at most
\(
    \xi\sqrt{k}\,\range(I),
\)
where \(\range(I)=\max I-\min I\).
\end{theorem}

Let us first consider an idealized setting. Suppose that \(s\) divides \(n\),
and that the array could be partitioned into \(n/s\) contiguous blocks,
which we call buckets, each containing exactly \(s\) cells. Suppose further
that the domain \([0,1]\) could be partitioned into \(n/s\) contiguous
intervals \(I_1,\ldots,I_{n/s}\), in such a way that bucket \(i\) receives
exactly the elements lying in \(I_i\). Then we could apply
\textsc{AdversarialPlacement} independently inside each bucket. The total
in-bucket cost would be at most
\[
    \xi\sqrt{s}\sum_{i=1}^{n/s}\range(I_i)
    =
    \xi\sqrt{s},
\]
since the intervals form a partition of \([0,1]\). Thus, if \(s\) were a
constant, the total in-bucket cost would also be constant. Moreover, because
the buckets are ordered consistently with the intervals, the boundary cost
between consecutive buckets is \(O(1)\): each boundary contribution is bounded
by the total length of the two adjacent intervals, and summing over all
boundaries gives at most a constant.

The difficulty, of course, is that such a partition is not available online. To use small buckets, the algorithm would need to know in advance which regions of the domain contain the right number of elements. In the stochastic version of online sorting, where the input values are sampled from a known distribution, a natural approach is to partition the domain into intervals of equal probability mass. This idea was introduced in~\cite{AbrESA}. However, even
with complete knowledge of the distribution, equal probability mass does not
imply balanced realized loads: random fluctuations may cause some buckets to
overflow while others remain underutilized. To absorb these fluctuations, the
algorithm of~\cite{AbrESA} uses buckets of polynomial size. Their analysis
therefore highlights a limitation of one-shot partitioning methods that commit
to the entire domain decomposition at the beginning of the execution.

Subsequent approaches~\cite{hu2025,KPT, hermansen26} overcame this limitation through
adaptivity. Their key insight is that the algorithm can observe fluctuations
as they occur and reallocate cells to buckets in response to past imbalances.
This makes it possible to use significantly smaller buckets, thereby reducing
the in-bucket cost, while incurring only a negligible increase in the
inter-bucket boundary cost. Our algorithm follows the same guiding principle:
rather than committing to a fine partition in advance, it gradually constructs
and adjusts the relevant buckets as information is revealed by the random
arrival order.

\subsection{The Algorithm and the Main Invariant}

We now give the formal statement of our main result for Random-Order Online
Sorting.

\begin{theorem}\label{thm:random-order-sorting}
There exists an absolute constant \(C>0\) such that, for every adversarially chosen multiset \(X\subseteq[0,1]\) of size \(n\),
\[
    \Pr\!\left[
        \cost(A)\le C\log^2 n\cdot \OPT(X)
    \right]
    \ge 1-n^{-4},
\]
where \(A\) is the array returned by Algorithm~\ref{alg:sorting} when the elements of \(X\) are revealed in uniformly random order.
\end{theorem}

We next describe the algorithm and the invariant that drives its analysis.
At a high level, the algorithm constructs a sequence of phases. Each phase
consists of a partition of \([0,1]\) into intervals together with a collection
of contiguous buckets of array cells associated with these intervals. Upon the
arrival of an element \(x\), the algorithm scans the buckets in creation order
and places \(x\) in the first bucket containing an empty cell eligible for \(x\). The purpose of the phase structure is to learn increasingly accurate
information about the rank structure of the input from the random prefixes
revealed so far, while continuing to place all arriving elements online. New
phases use empirical quantile partitions obtained from progressively larger
prefixes, whereas previously created phases retain enough remaining capacity to
absorb the discrepancies between these empirical partitions and the future
load.

The evolution of the phases has three parts. During the \emph{ascent}, the
algorithm repeatedly refines its empirical quantile partition as larger and
larger prefixes become available. This is the stage where the algorithm learns
the relevant rank-space structure online. At the finest scale, the algorithm
creates a middle phase with the same partition. It then enters the
\emph{descent}, in which consecutive quadruples of intervals are merged again.
Intuitively, the ascent creates increasingly accurate buckets, while the middle
and descent phases provide enough slack to absorb the load discrepancies caused
by sampling error.  The key invariant, stated below, is that at every update the buckets of the
immediately preceding phase are almost full, while the buckets of the phase
before that are completely full.

\begin{algorithm}[t]
\caption{\textsc{Random-Order Sorting}}
\label{alg:sorting}
\begin{algorithmic}[1]
\Require Empty array \(A\) of size \(n\), input sequence \(x_1,\ldots,x_n\).
\Ensure A placement of all elements into \(A\), or \textsc{Failure}.

\State Initialize the phase collection with \(\Phi_0\).
\State Compute the update times from Definition~\ref{def:time-grid}.

\For{\(t=1,\ldots,n\)}
    \State \(x\gets x_t\).
    \State Let \(B\) be the earliest created bucket containing an empty cell eligible for \(x\).
    \If{no such bucket exists}
        \State \Return \textsc{Failure}.
    \EndIf
    \State Place \(x\) in \(B\) using \textsc{AdversarialPlacement}.

    \If{\(t=t_j^{\uparrow}\) for some \(j\in\{1,\ldots,K\}\)}
        \State Create ascent phase \(\Phi_j^{\uparrow}\), or \Return \textsc{Failure}.
    \ElsIf{\(t=t^{\mathsf{mid}}\)}
        \State Create middle phase \(\Phi^{\mathsf{mid}}\), or \Return \textsc{Failure}.
    \ElsIf{\(t=t_j^{\downarrow}\) for some \(j\in\{1,\ldots,K\}\)}
        \State Create descent phase \(\Phi_j^{\downarrow}\), or \Return \textsc{Failure}.
    \EndIf
\EndFor

\State \Return \(A\).
\end{algorithmic}
\end{algorithm}

We now formalize the notation used by the algorithm. For simplicity of
presentation, we ignore rounding issues throughout the description. All update
times and bucket sizes should be rounded to integers in the natural way; this
affects only constant factors.

\begin{remark}\label{rem:ties}
Throughout the probabilistic analysis, we first assume that all input values
are distinct. This assumption is without loss of generality. If all values are
equal, then \(\OPT(X)=0\) and every placement has cost \(0\). Otherwise, let
\(\delta\) be the minimum positive difference between two distinct values of
\(X\), and perturb equal copies of each value by distinct infinitesimals of
magnitude at most \(\varepsilon\ll \delta/n\). This preserves the relative
order of unequal values and only imposes an arbitrary strict ordering among
equal copies.

The perturbed values are used only for rank comparisons, and hence only for
constructing the empirical partitions and determining to which bucket an
arriving element is assigned. Once a bucket has been selected, the element is
passed to \textsc{AdversarialPlacement} with its actual, unperturbed value.
Thus the perturbation serves solely as a consistent tie-breaking device and
does not affect the low-level placement routine or the cost of the resulting
array. Since \cref{thm:adversarial} applies to arbitrary input sequences,
including sequences with repeated values, its guarantee remains valid.
\end{remark}

\begin{definition}\label{def:time-grid}
Let \(D\ge 1\) and \(\beta\ge 1\) be sufficiently large absolute constants,
fixed throughout the analysis, and let
\[
    \eta = \frac{\beta}{\sqrt{\log n}}.
\]
Assume \(n\) is sufficiently large so that \(4\eta<1/10\); smaller values of
\(n\) can be absorbed into the constants. Let
\[
    K =
    \left\lfloor
        \log_4\left(\frac{n}{2D\log^2 n}\right)
    \right\rfloor .
\]
Define
\[
    \tau_0 = \frac n2,
    \qquad
    \tau_k = \frac{n}{2\cdot 4^k}
    \quad\text{for } k=1,\ldots,K,
    \qquad
    \tau_{K+1}=0.
\]
With this choice, $D\log^2 n \le \tau_K < 4D\log^2 n$. The ascent updates occur at times $
    t_j^{\uparrow} = (1-\eta)\tau_{K-j+1}$, for $\ j=1,\ldots,K.$
The middle update occurs at time
$
    t^{\mathsf{mid}} = (1-\eta)\tau_0.
$
The descent updates occur at times
$
    t_j^{\downarrow} = n-(1+4\eta)\tau_j,$ for $
    j=1,\ldots,K.
$
\end{definition}

\begin{definition}
A phase \(\Phi\) consists of a partition \(\mathcal I(\Phi)\) of \([0,1]\), together with a bucket \(B_\Phi(I)\) of array cells for every interval \(I\in\mathcal I(\Phi)\).
\end{definition}

Phases are ordered by creation time. When a new phase is created, its buckets are allocated as one contiguous block of fresh array cells, placed immediately after the buckets of all previously created phases. These cells remain assigned to that phase for the rest of the execution.

The guiding principle in constructing a new phase is to give every interval the
same effective capacity. The complication is that, when a new phase is created,
the preceding phase may still contain empty cells. Rather than discarding this
remaining capacity, we assign these cells to intervals of the new phase and take
them into account when determining the sizes of the new buckets.

We refer to this assignment as \emph{charging}. More precisely, when a new
phase \(\Phi'\) is created, each empty cell inherited from the preceding phase
is charged to some interval \(J\in\mathcal I(\Phi')\). We write \(c(J)\) for
the number of inherited empty cells charged to \(J\).

An empty cell that has never been charged is called \emph{ordinary}. Once an
empty cell is charged to an interval of a later phase, we call it
\emph{residual}. Residual cells remain physically located in the bucket in which
they were originally created, but are counted as inherited capacity when sizing
the buckets of the new phase. Each residual cell has an eligibility interval,
which determines which future arrivals may be placed into it. Thus, an ordinary empty cell is eligible for any arriving element if it falls into its bucket's interval. A residual empty cell (which carries over from a previous phase) is only eligible for elements falling into the specific interval assigned to it during the charging step.

How the eligibility interval is chosen depends on the evolution of the
partition. During the ascent, the new partition has more intervals, and on the
good event each old interval contains at least one interval of the new
partition. We charge the empty cells of the old bucket to one such contained
interval, and these residual cells are henceforth eligible only for elements
belonging to that interval. During the descent, several old intervals are
merged into a larger one. In this case, the inherited cells are charged to the
merged interval for the purpose of computing effective capacity, but retain
their original eligibility intervals. This preserves the commitments made when
the old buckets were created.

\begin{definition}[Balanced extension]\label{def:balanced-extension}
Suppose a new phase \(\Phi'\) with interval partition
\(\mathcal I(\Phi')\) is created, and suppose that the algorithm assigns
\(L\) fresh cells to its new buckets. Let \(c(J)\) denote the number of  residual cells charged to interval
\(J\in\mathcal I(\Phi')\). The bucket sizes \(s(J)\),
\(J\in\mathcal I(\Phi')\), are chosen so that the effective capacities
$
    s(J)+c(J)
$
are equal over all \(J\in\mathcal I(\Phi')\), and
\[
    \sum_{J\in\mathcal I(\Phi')}s(J)=L.
\]
Equivalently, if \(m=|\mathcal I(\Phi')|\), then
\[
    s(J)=\frac{L+\sum_{J'}c(J')}{m}-c(J).
\]
We say that the balanced extension is feasible if \(s(J)>0\) for every
\(J\in\mathcal I(\Phi')\).
\end{definition}

\begin{remark}
The quantities \(s(J)\) defined above need not be integers. To obtain integer
bucket sizes while preserving the total fresh-cell budget, we initially assign
each interval \(J\) a bucket of size \(\lfloor s(J)\rfloor\). Let
\[
    r
    =
    L-\sum_{J\in\mathcal I(\Phi')} \lfloor s(J)\rfloor.
\]
Since all charges \(c(J)\) are integers and the effective capacities
\(s(J)+c(J)\) are equal before rounding, we have
$
    0\le r<|\mathcal I(\Phi')|.
$
We then assign one additional fresh cell to any \(r\) intervals. Thus the total
number of fresh cells assigned is exactly \(L\), while the effective capacities
of any two intervals differ by at most \(1\).
\end{remark}

We now describe the phases created by the algorithm.

\paragraph{Initial phase.}
The algorithm begins with phase \(\Phi_0\). Its partition is
\(\mathcal I(\Phi_0)=\{[0,1]\}\), and its unique bucket has
\(\tau_K\) cells. Since \(\tau_K=\Theta(\log^2 n)\), the initial
bucket has \(O(\log^2 n)\) cells.

\paragraph{Ascent phases.}
For \(j=1,\ldots,K\), the \(j\)-th ascent update creates phase
\(\Phi^\uparrow_j\) at time \(t^\uparrow_j\). Let
$
    P_j=\{x_1,\ldots,x_{t^\uparrow_j}\}
$
be the prefix observed so far, and write its elements in increasing order as
$
    y_1<y_2<\cdots<y_{t^\uparrow_j}.
$
The new phase has \(4^j\) intervals. For \(i=0,\ldots,4^j\), define
$
    r_i^{(j)}
    =
    \left\lfloor
        {i\,t^\uparrow_j}/{4^j}
    \right\rfloor.
$
The empirical quantile partition is
$
    I_1=[0,y_{r_1^{(j)}}),
$
$
    I_i=
    [y_{r_{i-1}^{(j)}},y_{r_i^{(j)}})
    \text{ for }i=2,\ldots,4^j-1,
$
and
$
    I_{4^j}
    =
    [y_{r_{4^j-1}^{(j)}},1].
$
These intervals form \(\mathcal I(\Phi^\uparrow_j)\). By construction,
consecutive cutoffs satisfy
$
    r_i^{(j)}-r_{i-1}^{(j)}
    =
    {t^\uparrow_j}/{4^j}\pm1,
$
and hence every interval corresponds to
$
    {t^\uparrow_j}/{4^j}\pm O(1)
$
sample gaps.

The total number of fresh cells assigned to ascent phase \(j\) is
$
    L^\uparrow_j
    =
    \tau_{K-j}-\tau_{K-j+1}.
$
On the good event established in the analysis, every interval
\(I\in\mathcal I(\Phi^\uparrow_{j-1})\) contains at least one interval
\(J\in\mathcal I(\Phi^\uparrow_j)\), where
\(\Phi^\uparrow_0=\Phi_0\). We choose one such contained interval
\(J\subseteq I\) and charge all empty cells of
\(B_{\Phi^\uparrow_{j-1}}(I)\) to \(J\). These residual cells have eligibility
interval \(J\). The bucket sizes of \(\Phi^\uparrow_j\) are then chosen by the
balanced extension rule with budget \(L^\uparrow_j\).

\paragraph{Middle phase.}
The middle update creates phase \(\Phi^{\mathsf{mid}}\) at time
\(t^{\mathsf{mid}}\). Its interval partition is the same as that of the final
ascent phase:
$
    \mathcal I(\Phi^{\mathsf{mid}})
    =
    \mathcal I(\Phi_K^{\uparrow}).
$
The total number of new cells assigned to the middle phase is
$
    L^{\mathsf{mid}}=\tau_0-\tau_1=3n/8.
$
The charging map is the identity: for every
\(I\in\mathcal I(\Phi_K^{\uparrow})\), the empty cells of
\(B_{\Phi_K^{\uparrow}}(I)\) are charged to the corresponding interval
\(I\in\mathcal I(\Phi^{\mathsf{mid}})\), and their eligibility interval remains
\(I\). The bucket sizes are again chosen by the balanced extension rule.

\paragraph{Descent phases.}
For \(j=1,\ldots,K\), the \(j\)-th descent update creates phase
\(\Phi_j^{\downarrow}\) at time \(t_j^{\downarrow}\). The partition is obtained
by merging quadruples of intervals from the previous phase. More precisely,
let \(\Psi_{j-1}\) denote the previous phase, where
$
    \Psi_0=\Phi^{\mathsf{mid}}
    \text{ and }
    \Psi_{j-1}=\Phi_{j-1}^{\downarrow}
    \text{ for } j\ge 2.
$
If
$
    \mathcal I(\Psi_{j-1})
    =
    \{I_1,I_2,\ldots,I_{4\kappa}\}
$
in left-to-right order, then
$
    \mathcal I(\Phi_j^{\downarrow})
    =
    \{I'_1,\ldots,I'_{\kappa}\},
$
where
\[
    I'_r
    =
    I_{4(r-1)+1}\cup I_{4(r-1)+2}
    \cup I_{4(r-1)+3}\cup I_{4(r-1)+4}
\]
for $r=1,\ldots,\kappa$. The total number of new cells assigned to descent phase \(j\) is
$
    L_j^{\downarrow}=\tau_j-\tau_{j+1}.
$
Each old interval charges its empty cells to the unique new interval that
contains it, for the purpose of the balanced extension rule. However, these
residual cells keep their original eligibility interval. That is, if the empty
cells of \(B_{\Psi_{j-1}}(I_r)\) are charged to the merged interval \(I'_s\),
then they are counted in \(c(I'_s)\), but they remain eligible only for
elements lying in \(I_r\). The bucket sizes of \(\Phi_j^{\downarrow}\) are
then chosen by the balanced extension rule with budget \(L_j^{\downarrow}\).

The total number of cells allocated by the algorithm is exactly \(n\). Indeed,
\[
    \tau_K
    + \sum_{j=1}^K L_j^{\uparrow}
    + L^{\mathsf{mid}}
    + \sum_{j=1}^K L_j^{\downarrow}
    =
    \tau_K
    + (\tau_0-\tau_K)
    + (\tau_0-\tau_1)
    + \tau_1
    =
    n.
\]
The factor \(4\) in the refinement and merging steps is matched to the time
grid. Indeed, each phase is responsible for a time window whose length grows,
or shrinks, by a factor of \(4\) between consecutive scales. Since the number
of buckets changes by the same factor, the average number of cells per bucket
remains \(\Theta(\tau_K)=\Theta(\log^2 n)\) at every scale.

We are now ready to state the main invariant. The invariant is a deterministic
property of an execution of the algorithm. All statements referring to an update
time are interpreted after the element arriving at that time has been placed and
immediately before the new phase is created. Later, we prove that this property
holds with probability at least \(1-n^{-4}\) over the random arrival order.

\begin{definition}\label{def:main-invariant}
We say that an execution of Algorithm~\ref{alg:sorting} satisfies the
\emph{main invariant} if the following properties hold.

\begin{enumerate}
    \item For every \(j=1,\ldots,K\), immediately before the creation of
    ascent phase \(\Phi_j^{\uparrow}\), all buckets of
    \(\Phi_{j-2}^{\uparrow}\) are full, and every bucket of
    \(\Phi_{j-1}^{\uparrow}\) has between \(1\) and
    \(c_a\log^{3/2}n\) empty cells, where $c_a$ is a sufficiently large absolute constant. Here \(\Phi_0^{\uparrow}\) denotes the
    initial phase \(\Phi_0\), and the condition involving
    \(\Phi_{j-2}^{\uparrow}\) is ignored when \(j=1\).

    \item Immediately before the creation of the middle phase
    \(\Phi^{\mathsf{mid}}\), all buckets of
    \(\Phi_{K-1}^{\uparrow}\) are full, and every bucket of
    \(\Phi_K^{\uparrow}\) has between \(1\) and
    \(c_m\log^{3/2}n\) empty cells, where $c_m$ is a sufficiently large absolute constant.

    \item For every \(j=1,\ldots,K\), immediately before the creation of
    descent phase \(\Phi_j^{\downarrow}\), every bucket of the previous phase
    has between \(1\) and \(c_d\log^{3/2}n\) empty cells, and every bucket of
    the phase before that is full. More explicitly, for \(j=1\), all buckets
    of \(\Phi_K^{\uparrow}\) are full and every bucket of
    \(\Phi^{\mathsf{mid}}\) has between \(1\) and
    \(c_d\log^{3/2}n\) empty cells. For \(j\ge 2\), all buckets of
    \(\Phi_{j-2}^{\downarrow}\) are full and every bucket of
    \(\Phi_{j-1}^{\downarrow}\) has between \(1\) and
    \(c_d\log^{3/2}n\) empty cells, where $c_d$ is a sufficiently large absolute constant.

    \item Immediately before the final arrival \(x_n\) is placed, all buckets
created before \(\Phi_K^{\downarrow}\) are full, and the unique bucket of
\(\Phi_K^{\downarrow}\) has one empty cell.
\end{enumerate}
\end{definition}

\begin{definition}
We call an execution satisfying the main invariant a \emph{good run}, and denote by $\mathcal{G}$ the event that the execution of Algorithm~\ref{alg:sorting} is a good run.
\end{definition}

\subsection{The Main Invariant Implies Success }

We first show that every good run leads to a successful execution of
Algorithm~\ref{alg:sorting}. There are three possible ways in which the
algorithm may fail. First, during an ascent update, an interval of the previous
phase may contain no interval of the new empirical partition, in which case the
charging rule is not well-defined; Lemma~\ref{lem:ascent-containment} rules out
this possibility. Second, a balanced extension may be infeasible, assigning
nonpositive size to some newly created bucket; this is ruled out by
Lemma~\ref{lem:balanced-extension-feasibility}. Finally, an arriving element
may find no eligible empty cell in which it can be placed; Lemma~\ref{lem:no-placement-failure}
shows that this cannot occur either. Together, these lemmas imply that the
algorithm does not declare \textsc{Failure} on any good run.

Before addressing these three failure modes, we record an auxiliary consequence
of the time grid and the balanced extension rule that will be used repeatedly
throughout the analysis.

\begin{lemma}\label{lem:bucket-size-formula}
For every phase created after the initial phase, except for the final descent
phase \(\Phi_K^\downarrow\), if \(L\) denotes the total number of fresh cells
assigned to the phase and \(m\) denotes the number of intervals in its
partition, then
$
    L/m=3\tau_K/4.
$
Consequently, for every new interval \(J\) in such a phase,
\[
    s(J)=\frac34\tau_K+\frac{1}{m}\sum_{J'}c(J')-c(J).
\]
For the final descent phase \(\Phi_K^\downarrow\), the partition consists of
the single interval \([0,1]\), and its unique bucket has size $s([0,1])=\tau_K$. In particular, if every charge \(c(J)\) is \(O(\log^{3/2}n)\), then every
bucket created by the algorithm has size \(O(\log^2 n)\), and every bucket
created before the final descent phase has size
$
    3\tau_K/4 \pm O(\log^{3/2}n).
$
\end{lemma}

\begin{proof}
For an ascent phase \(j\),
\[
    \frac{L_j^\uparrow}{4^j}
    =
    \frac{\tau_{K-j}-\tau_{K-j+1}}{4^j}
    =
    \frac{3\tau_{K-j+1}}{4^j}
    =
    \frac34\tau_K.
\]
For the middle phase,
\[
    \frac{L^{\mathsf{mid}}}{4^K}
    =
    \frac{\tau_0-\tau_1}{4^K}
    =
    \frac34\tau_K.
\]
For a non-final descent phase \(j<K\),
\[
    \frac{L_j^\downarrow}{4^{K-j}}
    =
    \frac{\tau_j-\tau_{j+1}}{4^{K-j}}
    =
    \frac{3\tau_{j+1}}{4^{K-j}}
    =
    \frac34\tau_K.
\]
For the final descent phase \(j=K\), the partition has one interval and
$
    L_K^\downarrow=\tau_K-\tau_{K+1}=\tau_K.
$
Since there is only one interval, the balanced extension rule gives
\[
    s([0,1])
    =
    L_K^\downarrow+\sum_{J'}c(J')-c([0,1])
    =
    L_K^\downarrow
    =
    \tau_K.
\]

The formula for \(s(J)\) in all non-final phases is just
Definition~\ref{def:balanced-extension}. If every charge is
\(O(\log^{3/2}n)\), then the average charge is also
\(O(\log^{3/2}n)\), and so every non-final bucket has size
$
    3\tau_K/4 \pm O(\log^{3/2}n).
$
The final descent bucket has size \(\tau_K\). Since
\(\tau_K=\Theta(\log^2 n)\), the claim follows.
\end{proof}

\begin{lemma}\label{lem:ascent-containment}
Assume that an execution satisfies the main invariant. Then, for every
\(j=0,\ldots,K-1\) and every interval
\(I\in\mathcal I(\Phi_j^\uparrow)\), there exists an interval
\(J\in\mathcal I(\Phi_{j+1}^\uparrow)\) such that $J\subseteq I$.
Consequently, the charging rule in every ascent update is well-defined.
\end{lemma}

\begin{proof}
Fix \(j\in\{0,\ldots,K-1\}\) and an interval \(I\in\mathcal I(\Phi^\uparrow_j)\). Consider the time immediately before the creation of \(\Phi^\uparrow_{j+1}\).

By the main invariant, the bucket \(B_{\Phi^\uparrow_j}(I)\) has at most
\(O(\log^{3/2}n)\) empty cells at this time. If \(j=0\), then
\(B_{\Phi_0}(I)\) is the unique initial bucket, whose size is \(\tau_K\).
Hence it has received at least $\tau_K-O(\log^{3/2}n)$
elements. If \(j\ge1\), then the preceding lemma implies that every bucket
created in an ascent phase has size $ 3\tau_K/4\pm O(\log^{3/2}n), $ and therefore \(B_{\Phi^\uparrow_j}(I)\) has received at least $\frac34\tau_K-O(\log^{3/2}n)$ elements. Thus, in either case, by the time
\(\Phi^\uparrow_{j+1}\) is created, the prefix used to define its empirical
partition contains at least $3\tau_K/4-O(\log^{3/2}n)$
elements lying in \(I\).

Let $M=t^\uparrow_{j+1}, m=4^{j+1}, $
and recall that the empirical rank cutoffs are
$
    r_i^{(j+1)}
    =
    \left\lfloor{iM}/{m}\right\rfloor$ for $
    i=0,\ldots,m.
$
Define
$
    w_{j+1}
    =
    \max_{1\le i\le m}
    \left(
        r_i^{(j+1)}-r_{i-1}^{(j+1)}
    \right).
$
By the balanced-cutoff construction,
\[
    w_{j+1}
    \le
    \left\lceil\frac{M}{m}\right\rceil
    =
    \frac{1-\eta}{4}\tau_K+O(1).
\]
Moreover, every interval of
\(\mathcal I(\Phi^\uparrow_{j+1})\) contains at most \(w_{j+1}+1\)
points of the prefix: every interior interval contains at most
\(w_{j+1}\) such points, while the possible additional \(1\) accounts for
a boundary interval.

Suppose, for contradiction, that \(I\) contains no interval of
\(\mathcal I(\Phi^\uparrow_{j+1})\). Since both \(I\) and the intervals of
\(\mathcal I(\Phi^\uparrow_{j+1})\) are contiguous, \(I\) can then intersect
at most two intervals of the new partition; otherwise, an interval lying
between the first and last intersected intervals would be fully contained in
\(I\). Consequently, \(I\) contains at most
\[
    2w_{j+1}+2
    =
    \frac{1-\eta}{2}\tau_K+O(1)
\]
points of the prefix. On the other hand, as shown above, \(I\) contains at least
$
    3\tau_K/4-O(\log^{3/2}n)
$
points of the prefix. Since
$
    \tau_K=\Theta(\log^2 n),
$
for all sufficiently large \(n\),
\[
    \frac34\tau_K-O(\log^{3/2}n)
    >
    \frac{1-\eta}{2}\tau_K+O(1).
\]
This is a contradiction. Hence there exists an interval
$
    J\in\mathcal I(\Phi^\uparrow_{j+1})
$
such that \(J\subseteq I\). Therefore the charging rule in every ascent
update is well-defined.
\end{proof}

\begin{lemma}
\label{lem:balanced-extension-feasibility}
Assume that an execution satisfies the main invariant. Then every balanced
extension step assigns a positive number of fresh cells to every newly created
bucket.
\end{lemma}

\begin{proof}
Consider an update that creates a new phase \(\Phi'\). Let
\(m=|\mathcal I(\Phi')|\), let \(L\) be the number of fresh cells assigned to
the phase, and let \(c(J)\) be the charge assigned to
\(J\in\mathcal I(\Phi')\). By the balanced extension rule,
\[
    s(J)
    =
    \frac{L+\sum_{J'\in\mathcal I(\Phi')}c(J')}{m}
    -
    c(J)
    =
    \frac{L}{m}
    +
    \frac{1}{m}\sum_{J'\in\mathcal I(\Phi')}c(J')
    -
    c(J).
\]
Thus
$
    s(J)\ge {L}/{m}-c(J).
$

By the main invariant, every bucket whose empty cells are charged has at most
\(O(\log^{3/2}n)\) empty cells at the corresponding update time. In an ascent
update, Lemma~\ref{lem:ascent-containment} shows that the charging map is
well-defined, and each new interval receives the charge of at most one old
interval. The same is true in the middle update. In a descent update, each new
interval is the union of four old intervals, and hence receives charges from
at most four old intervals. Therefore, for every update and every new interval
\(J\),
$
    c(J)\le C\log^{3/2}n
$
for some absolute constant \(C\).

If \(\Phi'\neq \Phi_K^\downarrow\), then by
Lemma~\ref{lem:bucket-size-formula},
$
    {L}/{m}=3\tau_K/4=\Theta(\log^2 n).
$
Hence, for all sufficiently large \(n\),
$
    s(J)
    \ge
    3\tau_K/4-C\log^{3/2}n
    >0.
$ If \(\Phi'=\Phi_K^\downarrow\), then its partition consists of the single
interval \([0,1]\), and Lemma~\ref{lem:bucket-size-formula} gives
$
    s([0,1])=\tau_K>0.
$
Thus every balanced extension step assigns a positive number of fresh cells to
every newly created bucket.
\end{proof}

\begin{lemma}\label{lem:no-placement-failure}
Assume that an execution satisfies the main invariant. Then no arriving
element is rejected because all of its eligible empty cells are full.
\end{lemma}

\begin{proof}
Consider first a period between the creation of some phase \(\Phi\) and the
next update time. During this period, no later phase has yet been created, so
all empty cells of \(\Phi\) are ordinary cells. By the main invariant,
immediately before the next update, every bucket of \(\Phi\) has at least one
empty cell. Since the number of empty cells in a bucket is monotone
non-increasing over time, every bucket of \(\Phi\) has at least one ordinary
empty cell throughout the period.

Now fix an element \(x\) arriving during this period. Since
\(\mathcal I(\Phi)\) is a partition of \([0,1]\), there is an interval
\(I\in\mathcal I(\Phi)\) such that \(x\in I\). The bucket
\(B_\Phi(I)\) therefore contains an ordinary empty cell eligible for \(x\).
Hence the placement step cannot fail.

It remains to consider the final period after the creation of
\(\Phi_K^\downarrow\). This phase consists of the single interval \([0,1]\).
By item~(4) of the main invariant, immediately before the final arrival
\(x_n\) is placed, all buckets created before \(\Phi_K^\downarrow\) are full
and the unique bucket of \(\Phi_K^\downarrow\) has one empty cell. By
monotonicity, this bucket has at least one empty cell at every earlier time in
the final period. Since its interval is \([0,1]\), that empty cell is eligible
for every arriving element. Thus no placement failure occurs in the final
period either.
\end{proof}

Combining the preceding lemmas, we obtain the following deterministic
consequence.

\begin{corollary}\label{cor:invariant-implies-no-failure}
On the event \(\mathcal G\), Algorithm~\ref{alg:sorting} does not declare
\textsc{Failure}.
\end{corollary}

\subsection{Cost under a Good Run}

We now bound the cost of the algorithm on the event \(\mathcal G\). By
Corollary~\ref{cor:invariant-implies-no-failure}, on \(\mathcal G\) the
algorithm does not declare \textsc{Failure} and therefore returns a well-defined
array \(A\).

Let $R_X=[\min X,\max X]$. For an interval \(I\subseteq[0,1]\), write $\ell_X(I)=|I\cap R_X|$. Thus, for every partition \(\mathcal I\) of \([0,1]\),
\[
    \sum_{I\in\mathcal I}\ell_X(I)
    =
    \range(X)
    =
    \OPT(X).
\]

We first record a simple consequence of the balanced extension rule and the
main invariant. On \(\mathcal G\), every bucket created by the algorithm has
size \(O(\log^2 n)\). Indeed, the initial bucket has size
\(\Theta(\tau_K)=O(\log^2 n)\). For every later phase, the number of fresh
cells per interval is \((3/4)\tau_K=O(\log^2 n)\), while the total charge
assigned to any interval is \(O(\log^{3/2}n)\). Hence every bucket has size
\(O(\log^2 n)\).

\begin{lemma}\label{lem:cost-good-run}
On the event \(\mathcal G\), Algorithm~\ref{alg:sorting} returns an array \(A\)
satisfying
\[
    \cost(A)\le O(\log^2 n)\cdot\OPT(X).
\]
\end{lemma}

\begin{proof}
We decompose the cost into three parts.

\paragraph{In-bucket cost.}
The in-bucket cost is the contribution of adjacent cells that belong to the
same bucket. Consider a bucket \(B_\Phi(I)\) associated with an interval
\(I\in\mathcal I(\Phi)\). All elements placed in this bucket lie in
\(I\cap R_X\). By Theorem~\ref{thm:adversarial}, the total cost incurred
inside this bucket is at most
\[
    O\bigl(\sqrt{|B_\Phi(I)|}\bigr)\cdot \ell_X(I).
\]
Since every bucket has size \(O(\log^2 n)\), this is at most
$
    O(\log n)\cdot\ell_X(I).
$
Summing over all buckets of a fixed phase \(\Phi\), and using that
\(\mathcal I(\Phi)\) is a partition of \([0,1]\), we obtain
\[
    \sum_{I\in\mathcal I(\Phi)}
    O(\log n)\cdot\ell_X(I)
    =
    O(\log n)\cdot\OPT(X).
\]
The algorithm creates \(O(\log n)\) phases, so the total in-bucket cost is
$
    O(\log^2 n)\cdot\OPT(X).
$

\paragraph{Inter-bucket cost.}
The inter-bucket cost is the contribution of adjacent cells that belong to
different buckets of the same phase. Buckets are ordered consistently with
their associated intervals. Thus, if two neighboring buckets are associated
with neighboring intervals \(I\) and \(J\), then the edge crossing between
these buckets has length at most
$
    \ell_X(I)+\ell_X(J).
$
Each interval has at most two neighboring intervals, and hence each
\(\ell_X(I)\) is charged only a constant number of times. Therefore, the total
inter-bucket cost within a fixed phase is
$
    O(\OPT(X)).
$
Summing over the \(O(\log n)\) phases gives total inter-bucket cost
$
    O(\log n)\cdot\OPT(X).
$

\paragraph{Inter-phase cost.}
Finally, we consider adjacent cells that belong to different phases. Each such
edge has length at most \(\OPT(X)\), since all input values lie in \(R_X\).
There are \(O(\log n)\) phases, and therefore only \(O(\log n)\) phase
boundaries. Hence the total inter-phase cost is
$
    O(\log n)\cdot\OPT(X).
$

Combining the three bounds gives
\[
    \cost(A)
    \le
    O(\log^2 n)\cdot\OPT(X).
\]
\end{proof}

\subsection{Probability of a Good Run}

It remains to show that the main invariant holds with high probability. More
precisely, we prove the following stronger statement: for every fixed constant
\(\alpha>0\), after choosing the constants \(D\) and \(\beta\) in
Definition~\ref{def:time-grid} sufficiently large as functions of \(\alpha\),
$
    \Pr[\mathcal G]\ge 1-n^{-\alpha}.
$
Taking \(\alpha=4\), together with the results of the preceding subsections,
will prove Theorem~\ref{thm:random-order-sorting}.

The proof has three steps. First, we represent the interaction between an
observed prefix and a later time window as a uniformly random binary string.
Second, we use this representation to show that every empirical rank interval
receives close to its expected load. Finally, we show that all the concentration estimates required throughout the execution hold simultaneously with high probability. On this event, called \(\mathcal E_{\mathrm{reg}}\),  every phase receives the right amount of load over the relevant time windows, which in turn implies the main invariant, i.e. $\mathcal E_{\mathrm{reg}} \subseteq \mathcal{G}$.

\paragraph{Random-order sampling and rank gaps.}

We begin with the elementary sampling property underlying the analysis.

\begin{proposition}\label{prop:random-order-sampling}
Let \(U\) be a set of \(N\) distinct elements revealed in uniformly random
order \(u_1,\ldots,u_N\). Fix \(M\in[N]\) and a deterministic set of time
indices
$
    W\subseteq\{M+1,\ldots,N\},\ |W|=L
$. Define
\[
    S=\{u_1,\ldots,u_M\},
    \qquad
    F=\{u_t:t\in W\},
    \qquad
    H=S\cup F.
\]
Then \(H\) is a uniformly random \((M+L)\)-subset of \(U\). Moreover,
conditioned on \(H\), the prefix sample \(S\) is a uniformly random
\(M\)-subset of \(H\).
\end{proposition}

\begin{proof}
All \((M+L)\)-subsets of \(U\) are equally likely to occupy the positions
\([M]\cup W\), so \(H\) is uniform. Conditioned on \(H\), all assignments of
the elements of \(H\) to these positions are equally likely. Hence the subset
occupying the first \(M\) positions is a uniformly random \(M\)-subset of
\(H\).
\end{proof}

Thus, after conditioning on the points appearing in the prefix and in a later
time window, the only remaining randomness is a uniformly random split of
these points into \(M\) sample points and \(L\) window points. The following
definition encodes this split in rank space.

\begin{definition}[Gaps]\label{def:gaps}
Fix a set \(H\) of \(M+L\) distinct points and a subset \(S\subseteq H\) of
size \(M\). Sort the elements of \(H\) in increasing order, encode every point
of \(S\) by \(1\), and encode every point of \(H\setminus S\) by \(0\).
Suppose the ones occur at positions
$
    1\le i_1<i_2<\cdots<i_M\le M+L.
$
The associated gaps are
$
    g_0=i_1-1,
$
$ g_r=i_{r+1}-i_r-1$ for $r=1,\ldots,M-1,$ and $g_M=M+L-i_M.$
Thus \(g_r\) counts the number of non-sample points between consecutive sample
points, with \(g_0\) and \(g_M\) denoting the two boundary gaps. In particular,
\[
    \sum_{r=0}^M g_r=L.
\]
\end{definition}

\begin{lemma}\label{lem:exchangeable-gaps}
In the setting of Proposition~\ref{prop:random-order-sampling}, condition on
the set \(H\). Then the gap vector
$
    (g_0,g_1,\ldots,g_M)
$
is uniformly distributed over all weak compositions of \(L\) into \(M+1\)
parts. In particular, it is exchangeable.
\end{lemma}

\begin{proof}
Conditioned on \(H\), the sample \(S\) is a uniformly random \(M\)-subset of
\(H\). Equivalently, the positions
$
    1\le i_1<i_2<\cdots<i_M\le M+L
$
of the sample points form a uniformly random \(M\)-subset of
\(\{1,\ldots,M+L\}\). The map
\[
    (i_1,\ldots,i_M)\mapsto(g_0,\ldots,g_M)
\]
is the standard stars-and-bars bijection between \(M\)-subsets of
\([M+L]\) and weak compositions of \(L\) into \(M+1\) parts. Hence the gap
vector is uniform over all such compositions. Since this distribution is
invariant under permutations of its coordinates, the gaps are exchangeable.
\end{proof}

The empirical intervals used by the algorithm correspond to blocks of
consecutive sample gaps. By exchangeability, the distribution of the number of
window points in such a block depends only on the number of gaps it contains,
and not on its position in rank space.

\begin{lemma}\label{lem:gap-sum-concentration}
For every constant \(C_0>0\), there exists a constant \(c>0\) such that the
following holds. Let \((g_0,\ldots,g_M)\) be the gap vector induced by a
uniformly random \(M\)-subset of a set of size \(M+L\), where
$
    L\le C_0M.
$
For every set \(Q\subseteq\{0,\ldots,M\}\) of size \(q\ge1\), define
\[
    G_Q=\sum_{r\in Q}g_r,
    \qquad
    \mu_Q=\frac{Lq}{M+1}.
\]
Then, for every \(a>0\),
\[
    \Pr\!\left[
        |G_Q-\mu_Q|\ge a
    \right]
    \le
    2\exp\left(
        -c\min\left\{
            \frac{a^2}{\mu_Q+1},
            a
        \right\}
    \right).
\]
\end{lemma}

We defer the proof of this Lemma to the Appendix~\ref{App:ConcentrationProof}.

\begin{corollary}\label{cor:gap-sum-high-probability}
For every \(B,C_0>0\), there exist constants \(c_0,C>0\),
depending only on \(B\) and \(C_0\), such that the following holds. Under the
assumptions of Lemma~\ref{lem:gap-sum-concentration}, if
$
    \mu_Q\ge c_0\log^2 n,
$
then
\[
    \Pr\!\left[
        |G_Q-\mu_Q|
        \ge C\sqrt{\mu_Q\log n}
    \right]
    \le n^{-B}.
\]
\end{corollary}

\begin{proof}
Apply \cref{lem:gap-sum-concentration} with
$
    a=C\sqrt{\mu_Q\log n}.
$
If \(\mu_Q\ge c_0\log^2 n\), then, after choosing \(c_0\) sufficiently
large, \(a\le\mu_Q\). Hence
\[
    \min\left\{
        \frac{a^2}{\mu_Q+1},
        a
    \right\}
    \ge C'\log n,
\]
where \(C'\) can be made arbitrarily large by increasing \(C\).
Choosing \(C\) sufficiently large as a function of \(B\) and \(C_0\)
gives the claim.
\end{proof}

\paragraph{Regularity of empirical rank intervals.}

We now translate the gap concentration estimate into the form used by the
algorithm. Importantly, the window whose load we estimate need not immediately
follow the sample prefix.

Fix a set \(U\) of \(n\) distinct points revealed in uniformly random order
$
    u_1,\ldots,u_n.
$
For integers \(M,L\), let
$
    S=\{u_1,\ldots,u_M\}
$
be the prefix sample, and let
$
    W\subseteq\{M+1,\ldots,n\},
    \ |W|=L,
$
be a deterministic future time window. Write the points of \(S\) in increasing
order as
$
    y_1<y_2<\cdots<y_M.
$

Consider an interval \(J\) whose endpoints are determined by fixed sample
ranks; that is, \(J\) is of the form
$
    [y_a,y_b),
$
up to the two boundary cases, where the rank indices \(a,b\) are fixed in
advance. Let \(Q(J)\) denote the corresponding set of consecutive sample gaps,
and define
\[
    q(J)=|Q(J)|,
    \qquad
    \mu(J)=\frac{Lq(J)}{M+1}.
\]
Finally, let \(N(W,J)\) denote the number of arrivals with time indices in
\(W\) whose values lie in \(J\).

\begin{lemma}\label{lem:rank-window-regularity}
For every \(\alpha,C_0>0\), there exist constants \(c_0,C>0\) such
that the following holds. Suppose that
$
    L\le C_0M,
$
and let \(\mathcal J\) be any family of at most \(n\) intervals of the form
described above. If
$
    \mu(J)\ge c_0\log^2 n
    \text{ for every }J\in\mathcal J,
$
then, with probability at least \(1-n^{-\alpha}\), simultaneously for every
\(J\in\mathcal J\),
\[
    |N(W,J)-\mu(J)|
    \le
    C\sqrt{\mu(J)\log n}.
\]
\end{lemma}

\begin{proof}
Fix \(J\in\mathcal J\), and let \(Q=Q(J)\). Define
\[
    F=\{u_t:t\in W\},
    \qquad
    H=S\cup F.
\]
By Proposition~\ref{prop:random-order-sampling}, \(H\) is a uniformly random
\((M+L)\)-subset of \(U\), and, conditioned on \(H\), the prefix sample \(S\)
is a uniformly random \(M\)-subset of \(H\).

Sort the elements of \(H\), encode the points of \(S\) by ones and the points
of \(F\) by zeros, and let
$
    (g_0,\ldots,g_M)
$
be the resulting gap vector. By Lemma~\ref{lem:exchangeable-gaps}, this vector
is uniformly distributed over the weak compositions of \(L\) into \(M+1\)
parts. Moreover,
\[
    N(W,J)=\sum_{r\in Q}g_r.
\]
Applying \cref{cor:gap-sum-high-probability} with
\(B=\alpha+2\), we obtain
\[
    \Pr\!\left[
        |N(W,J)-\mu(J)|
        >
        C\sqrt{\mu(J)\log n}
    \right]
    \le
    n^{-(\alpha+2)}
\]
for a sufficiently large constant \(C\). Since
$
    |\mathcal J|\le n,
$
a union bound over all \(J\in\mathcal J\) gives total failure probability at
most
$
    n^{-(\alpha+1)}
    \le
    n^{-\alpha}.
$
\end{proof}

\paragraph{The global regularity event.}

We now collect the regularity estimates needed in the invariant proof. It is
convenient to index all phases in their order of creation. Let
\[
    \Psi_0=\Phi_0,
    \qquad
    \Psi_j=\Phi^\uparrow_j
    \quad (1\le j\le K),
\]
\[
    \Psi_{K+1}=\Phi^{\mathrm{mid}},
    \qquad
    \Psi_{K+1+j}=\Phi^\downarrow_j
    \quad (1\le j\le K).
\]
Let \(u_r\) denote the creation time of \(\Psi_r\), so that
\[
    u_0=0,
    \qquad
    u_j=t^\uparrow_j
    \quad (1\le j\le K),
\]
\[
    u_{K+1}=t^{\mathrm{mid}},
    \qquad
    u_{K+1+j}=t^\downarrow_j
    \quad (1\le j\le K).
\]
For each phase \(\Psi_r\), let \(M_r\) denote the prefix length whose order
statistics determine its partition:
\[
    M_r=
    \begin{cases}
        u_r, & 1\le r\le K,\\
        u_K, & K+1\le r\le 2K+1.
    \end{cases}
\]
Thus every ascent partition is defined by its own prefix, whereas the middle
and descent partitions are coarsenings of the final ascent partition.

For \(1\le r\le2K\), define the one-step window
$
    W_r^{(1)}
    =
    \{u_r+1,\ldots,u_{r+1}\}.
$
Finally, define
$
    W_{\mathrm{fin}}
    =
    \{u_{2K+1}+1,\ldots,n-1\}
$. For a prefix length \(M\), a deterministic future window
$
    W\subseteq\{M+1,\ldots,n\},
$
and a family \(\mathcal J\) of intervals defined by deterministic rank
endpoints of the prefix, let \(q_M(J)\) denote the number of prefix gaps
corresponding to \(J\), and set
\[
    \mu_M(W,J)
    =
    \frac{|W|q_M(J)}{M+1}.
\]

Fix \(\alpha>0\). Apply Lemma~\ref{lem:rank-window-regularity} with failure
exponent \(\alpha+2\) and with a sufficiently large absolute constant \(C_0\)
that dominates the ratio between the length of every relevant window and the
corresponding prefix length. Let \(c_0\) and \(C_{\mathrm{reg}}\) be the
resulting constants. We write
$
    \operatorname{Reg}(M,W,\mathcal J)
$
for the event that, simultaneously for every \(J\in\mathcal J\),
\[
    |N(W,J)-\mu_M(W,J)|
    \le
    C_{\mathrm{reg}}
    \sqrt{\mu_M(W,J)\log n}.
\]
The global regularity event \(\mathcal E_{\mathrm{reg}}\) is the intersection
of the following events:

\begin{enumerate}
    \item[\textnormal{(R1)}]
    For every \(1\le r\le2K\),
    \[
        \operatorname{Reg}
        \bigl(
            M_r,
            W_r^{(1)},
            \mathcal I(\Psi_r)
        \bigr).
    \]
    These events control the load of each phase until the next update. During
    the ascent and middle phase, they also ensure that residual cells inherited
    from the preceding phase are filled during this window.

    \item[\textnormal{(R2)}]
    For every \(K+2\le r\le2K\),
    \[
        \operatorname{Reg}
        \bigl(
            M_{r-1},
            W_r^{(1)},
            \mathcal I(\Psi_{r-1})
        \bigr).
    \]
    These events control the original child intervals of residual cells
    inherited during the descent.

    \item[\textnormal{(R3)}]
    For the final descent phase,
    \[
        \operatorname{Reg}
        \bigl(
            M_{2K},
            W_{\mathrm{fin}},
            \mathcal I(\Psi_{2K})
        \bigr).
    \]
    This is the corresponding child-interval regularity event for the final
    period.
\end{enumerate}






\begin{lemma}\label{lem:global-regularity}
For every fixed \(\alpha>0\), after choosing \(D\) sufficiently large as a
function of \(\alpha\), the event \(\mathcal E_{\mathrm{reg}}\) satisfies
$
    \Pr[\mathcal E_{\mathrm{reg}}]
    \ge
    1-n^{-\alpha}.
$
\end{lemma}

\begin{proof}
Every event appearing in the definition of
\(\mathcal E_{\mathrm{reg}}\) is an application of
Lemma~\ref{lem:rank-window-regularity}. By the definition of \(C_0\), every
relevant window satisfies
$
    |W|\le C_0M,
$
where \(M\) is the prefix length defining the corresponding partition. Moreover, direct inspection of the time grid shows that, for every
interval-window pair appearing in \textnormal{(R1)}--\textnormal{(R3)},
$
    \mu_M(W,J)=\Theta(\tau_K),
$
with absolute implicit constants. Since
$
    \tau_K=\Theta(D\log^2 n),
$
choosing \(D\) sufficiently large guarantees
$
    \mu_M(W,J)\ge c_0\log^2 n
$
for every relevant pair.

Hence every regularity event in
\textnormal{(R1)}--\textnormal{(R3)} fails with probability at most
$
    n^{-(\alpha+2)}.
$
There are \(O(K)=O(\log n)\) such events. Therefore, by a union bound,
\[
    \Pr[\mathcal E_{\mathrm{reg}}^c]
    \le
    O(\log n)n^{-(\alpha+2)}
    \le
    n^{-\alpha}
\]
for all sufficiently large \(n\).
\end{proof}

We now show that the global regularity event implies a good run. For the
remainder of the subsection, set
$
    b=3\tau_K/4,
    \ 
    \Delta=\eta\tau_K.
$
Since every relevant expectation is \(\Theta(\tau_K)\), there exists a
constant \(C_\Gamma>0\), depending only on \(C_{\mathrm{reg}}\) and the
absolute constants in the time grid, such that all deviations guaranteed by
\(\mathcal E_{\mathrm{reg}}\) are at most
$
    \Gamma
    :=
    C_\Gamma\sqrt{\tau_K\log n}.
$
By choosing \(D\) and \(\beta\) sufficiently large, we may assume throughout
that
$
    \Gamma\le{\Delta}/{100}
    \text{ and }
    \eta\le 1/{100}.
$
All \(O(1)\) errors caused by the suppressed rounding of update times and rank
widths are absorbed into \(\Gamma\).

\begin{lemma}\label{lem:regularity-implies-invariant}
We have
$
    \mathcal E_{\mathrm{reg}}
    \subseteq
    \mathcal G.
$
\end{lemma}

\begin{proof}[Proof sketch]
We prove the main invariant inductively over the phase updates. At each update,
the immediately preceding phase has only
$
    \Delta\pm O(\Gamma)
$
empty cells in each bucket. These cells become residual and are charged to
intervals of the newly created phase.

During the ascent and the middle phase, every residual group has the same
eligibility interval as some interval of the new phase. By
\textnormal{(R1)}, this interval receives \(\Omega(\tau_K)\) arrivals before
the next update, while the residual group contains only
$
    \Delta+O(\Gamma)=o(\tau_K)
$
cells. Since the algorithm always places an arriving element into the earliest
created bucket containing an eligible empty cell, all such residual cells are
filled before later-created fresh cells with the same eligibility interval are
used. During the descent, the corresponding conclusion follows from the
child-interval regularity events \textnormal{(R2)}, while the final period is
handled by \textnormal{(R3)}.

At the same time, the one-step regularity estimates determine the load received
by every interval of the newly created phase. The time grid is chosen so that
the effective capacity per interval is
$
    b+\frac{\Delta}{4},
    \ 
    b+\Delta,
    \text{ or }
    b+4\Delta,
$
for ascent, middle, and descent phases, respectively, where
$
    b=3\tau_K/4
    \text{ and }
    \Delta=\eta\tau_K.
$
The corresponding one-step loads are
$
    (1-\eta)b\pm O(\Gamma),
    \ 
    b\pm O(\Gamma),
    \ 
    (1+4\eta)b\pm O(\Gamma).
$
Using
$
    \eta b=3\Delta/4
    \text{ and }
    4\eta b=3\Delta,
$
we obtain in every case
$
    \Delta\pm O(\Gamma)
$
empty cells in each bucket of the new phase at the next update.

Thus, one update later, the new phase is almost full, while the preceding
phase is completely full. Choosing \(D\) and \(\beta\) sufficiently large
ensures
$
    \Gamma\ll\Delta=O(\log^{3/2}n),
$
and therefore all clauses of the main invariant hold. The full verification is
deferred to \cref{app:regularity-implies-invariant}.
\end{proof}

We can now prove Theorem~\ref{thm:random-order-sorting}.

\begin{proof}[Proof of Theorem~\ref{thm:random-order-sorting}]
Take \(\alpha=4\). By Lemma~\ref{lem:global-regularity},
$
    \Pr[\mathcal E_{\mathrm{reg}}]
    \ge
    1-n^{-4}.
$
By Lemma~\ref{lem:regularity-implies-invariant},
$
    \mathcal E_{\mathrm{reg}}
    \subseteq
    \mathcal G.
$
On \(\mathcal G\), Corollary~\ref{cor:invariant-implies-no-failure} guarantees
that Algorithm~\ref{alg:sorting} returns a well-defined array, and
Lemma~\ref{lem:cost-good-run} gives
$
    \cost(A)
    \le
    O(\log^2 n)\cdot\OPT(X).
$
Therefore,
\[
    \Pr\!\left[
        \cost(A)
        \le
        O(\log^2 n)\cdot\OPT(X)
    \right]
    \ge
    1-n^{-4},
\]
which proves the theorem.
\end{proof}

\section{Random-Order Online TSP }

We will use the following parameterized form of Theorem \ref{thm:random-order-sorting}, in which the
failure probability and competitive ratio are measured with respect to an
ambient parameter \(N\), which may be larger than the size of the sorting
instance.

\begin{theorem}\label{thm:sorting-parameterized}
For every fixed constant \(A>0\), there exist constants \(C_A,D_A>0\) such
that the following holds. Let \(N\ge 2\), and let \(m\le N\) satisfy
$
    m\ge D_A\log^2 N.
$
Then there is an online algorithm for every random-order online sorting
instance \(Y\) of size \(m\) such that, with probability at least
\(1-N^{-A}\), the algorithm returns an ordering of cost at most
$
    C_A\log^2 N\cdot \OPT(Y).
$

\end{theorem}

\begin{proof}
The proof is the same as that of Theorem \ref{thm:random-order-sorting}, with the parameters of the
algorithm chosen as functions of the ambient parameter \(N\). More precisely,
set
$
    \eta={\beta}/{\sqrt{\log N}}
$
and choose the final time scale \(\tau_K\) so that
$
    D\log^2 N\le \tau_K<4D\log^2 N.
$
The assumption \(m\ge D_A\log^2 N\), for a sufficiently large constant
\(D_A\), guarantees that the resulting time grid is well-defined.

Every concentration argument in Section~\ref{sec: Sorting} then has expectation
\(\Theta(\log^2 N)\) and error equal to
$
    O\!\left(\sqrt{\tau_K\log N}\right)
    =O(\log^{3/2}N).
$
By choosing \(D\) and \(\beta\) sufficiently large as functions of \(A\),
Lemma~\ref{lem:rank-window-regularity} gives failure probability at most \(N^{-(A+2)}\) for each relevant
regularity event. Since there are \(O(\log m)\le O(\log N)\) such events, the
same union-bound argument as in Lemma~\ref{lem:global-regularity} shows that the main invariant holds
with probability at least \(1-N^{-A}\).

On this event, every bucket has size \(O(\log^2 N)\). Hence the in-bucket cost
of each phase is \(O(\log N)\OPT(Y)\), and there are
\(O(\log m)\le O(\log N)\) phases. The proof of Lemma~\ref{lem:cost-good-run} therefore gives
total cost
$
    O(\log^2 N)\OPT(Y).
$
Absorbing the constants into \(C_A\) proves the theorem.
\end{proof}

Our approach is to reduce the multidimensional problem to the one-dimensional machinery developed earlier. The reduction is based on projecting points onto a suitable curve. Let \(\Gamma\subseteq [0,1]^d\) be a rectifiable curve, and let $\gamma:[0,\ell(\Gamma)]\to [0,1]^d$ be an arc-length parameterization of \(\Gamma\). For each point \(x\in [0,1]^d\), let
\[
    \pi_\Gamma(x)\in \arg\min_{y\in\Gamma}\|x-y\|
\]
be an arbitrary nearest point on \(\Gamma\), and let
$
    \theta_\Gamma(x)\in [0,\ell(\Gamma)]
$
be such that \(\gamma(\theta_\Gamma(x))=\pi_\Gamma(x)\). Thus
\(\theta_\Gamma(x)\) is the one-dimensional key of \(x\) induced by its
projection onto \(\Gamma\).

The usefulness of this projection is captured by the following deterministic
observation.

\begin{lemma}\label{lem:projection}
Let \(Y\subseteq [0,1]^d\) be a set of points, and let \(\Gamma\) be a rectifiable curve. Suppose an ordering \(\sigma\) of the points of \(Y\) has one-dimensional projection cost
\[
    \sum_{i=2}^{|Y|}
    \left|
        \theta_\Gamma(y_{\sigma(i)})
        -
        \theta_\Gamma(y_{\sigma(i-1)})
    \right|
    \le L.
\]
Then the cost of the same ordering in the original metric is at most
\[
    L+2\sum_{y\in Y}d(y,\Gamma).
\]
\end{lemma}

\begin{proof}
For every consecutive pair in the ordering, the triangle inequality gives
\[
    \|y_{\sigma(i)}-y_{\sigma(i-1)}\|
    \le
    d(y_{\sigma(i)},\Gamma)
    +
    \|\pi_\Gamma(y_{\sigma(i)})-\pi_\Gamma(y_{\sigma(i-1)})\|
    +
    d(y_{\sigma(i-1)},\Gamma).
\]
Since \(\gamma\) is parameterized by arc length,
\[
    \|\pi_\Gamma(y)-\pi_\Gamma(y')\|
    \le
    |\theta_\Gamma(y)-\theta_\Gamma(y')|
\]
for every pair \(y,y'\). Summing over all consecutive pairs gives
\[
    \sum_{i=2}^{|Y|}
    \|y_{\sigma(i)}-y_{\sigma(i-1)}\|
    \le
    L+2\sum_{y\in Y}d(y,\Gamma),
\]
as claimed.
\end{proof}

We now describe the algorithm. Let $m_0=C_0\log^2 n$ for a sufficiently large constant \(C_0\), rounded to an integer. If \(n=O(m_0)\), the algorithm places all points arbitrarily; this incurs cost \(O(\log^2 n)\OPT(X)\), which is dominated by the claimed bound. Thus, assume below that \(n\) is sufficiently large compared to \(m_0\).

The algorithm first places the initial prefix of \(m_0\) points arbitrarily. It then proceeds in phases. At the beginning of a phase, suppose the first \(M\) points have already been placed. If more than \(m_0\) points remain, the algorithm handles the next $L=\min\{M,n-M\}$ points. Thus \(L\le M\), and except for the last possible leftover block, the prefix size doubles after each phase. If at most \(m_0\) points remain, the algorithm places them arbitrarily and stops.

\begin{algorithm}
\caption{\textsc{Random-Order Online TSP}}
\label{alg:tsp}
\begin{algorithmic}[1]
\Require Empty array \(A\) of size \(n\), input stream \(x_1,\ldots,x_n\).
\Ensure A placement of all points into \(A\), or \textsc{Failure}.

\State Let \(m_0=C_0\log^2 n\), rounded to an integer.
\If{\(n\le 2m_0\)}
    \State Place all points arbitrarily in \(A\).
    \State \Return \(A\).
\EndIf

\State Place the first \(m_0\) arriving points arbitrarily in
\(A[1],\ldots,A[m_0]\).
\State \(M\gets m_0\).

\While{\(n-M>m_0\)}
    \State \(L\gets \min\{M,n-M\}\).
    \State \(P\gets \{x_1,\ldots,x_M\}\).
    \State Compute a shortest Hamiltonian path \(\Gamma\) through \(P\).
    \State Let \(\gamma:[0,\ell(\Gamma)]\to[0,1]^d\) be an arc-length
    parameterization of \(\Gamma\).
    \State Initialize an instance of
    \textsc{Random-Order Sorting} on the block
    \(A[M+1],\ldots,A[M+L]\).

    \For{\(t=M+1,\ldots,M+L\)}
        \State Let \(\pi_\Gamma(x_t)\) be a nearest point to \(x_t\) on
        \(\Gamma\), and let
        \(\theta_\Gamma(x_t)\in[0,\ell(\Gamma)]\) be its position along
        \(\Gamma\).
        \If{\(\ell(\Gamma)>0\)}
            \State \(z_t\gets \theta_\Gamma(x_t)/\ell(\Gamma)\).
        \Else
            \State \(z_t\gets 0\).
        \EndIf
        \State Feed \(z_t\) to the one-dimensional sorting instance.
        \If{the sorting instance declares \textsc{Failure}}
            \State \Return \textsc{Failure}.
        \EndIf
        \State Place \(x_t\) in the corresponding cell of
        \(A[M+1],\ldots,A[M+L]\).
    \EndFor

    \State \(M\gets M+L\).
\EndWhile

\If{\(M<n\)}
    \State Place the remaining points \(x_{M+1},\ldots,x_n\) arbitrarily in
    the remaining cells.
\EndIf

\State \Return \(A\).
\end{algorithmic}
\end{algorithm}

We do not optimize computational efficiency here and allow the algorithm to
compute an exact shortest Hamiltonian path. Alternatively, one may use any
constant-factor Hamiltonian path through the prefix, changing only the
constants in the final competitive ratio.

We first record the easy length bound on the reference curve.

\begin{proposition}\label{prop:prefix-curve-length}
For every phase with prefix \(P\) and reference curve \(\Gamma\),
$
    \ell(\Gamma)\le \OPT(X).
$
\end{proposition}

\begin{proof}
Since \(P\subseteq X\), an optimal ordering of all points in \(X\) induces, by
shortcutting, an ordering of the points of \(P\) of no larger cost. Since
\(\Gamma\) is a shortest Hamiltonian path through \(P\), we have
$
    \ell(\Gamma)=\OPT(P)\le \OPT(X).
$
\end{proof}

It remains to control the distance of the current phase from the prefix curve.
This is where the random-order assumption is used.

\begin{lemma}\label{lem:projection-distance}
Fix a phase with prefix $P=\{x_1,\ldots,x_M\}$ and current block $F=\{x_{M+1},\ldots,x_{M+L}\},$ where \(L\le M\). Let \(\Gamma\) be the shortest Hamiltonian path through
\(P\). For every \(A>0\), there is a constant \(C_A>0\) such that
\[
    \Pr\left[
        \sum_{x\in F}d(x,\Gamma)
        >
        C_A\log n\cdot \OPT(X)
    \right]
    \le
    2n^{-(A+2)}.
\]
\end{lemma}

\begin{proof}

Since every point of \(P\) lies on \(\Gamma\), we have
$ d(x,\Gamma)\le d(x,P)$ 
for every $x\in F$. Thus it suffices to bound
$
    \sum_{x\in F}d(x,P).
$ Because the input \(X\) is a multiset, we regard identical points arriving at
different times as distinct labeled occurrences. Condition on the collection of
occurrences
$
    Z=P\cup F=\{x_1,\ldots,x_{M+L}\},
$
where the labels distinguish coincident copies. Coincident points simply have
distance zero and do not affect any of the metric bounds below. Since the
arrival order is uniformly random, conditioned on \(Z\), the prefix \(P\) is a
uniformly random \(M\)-subset of the \(M+L\) labeled points in \(Z\), and
\(F\) is the complementary \(L\)-subset. Equivalently, the points in
\(Z\) are colored red and blue uniformly at random, with exactly \(M\) red
points and \(L\) blue points; the red points are \(P\), and the
blue points are \(F\).

Let \(T\) be a shortest Hamiltonian cycle through the points of \(Z\), and let
\(L_T=\ell(T)\). Since \(Z\subseteq X\), an optimal Hamiltonian path on \(X\)
can be shortcut to a Hamiltonian path on \(Z\) of length at most \(\OPT(X)\).
Closing this path into a cycle adds at most another \(\OPT(X)\). Hence
$
    L_T\le 2\OPT(X).
$

Traverse \(T\) cyclically. The blue points appear in maximal blue runs between
consecutive red points. Consider such a run. Let \(b\) be the number of blue
points in the run, and let \(\ell\) be the length of the corresponding
red-to-red subpath of \(T\). For every blue point \(x\) in the run, the two
path distances from \(x\) to the red endpoints sum to \(\ell\), so the shorter
one is at most \(\ell/2\). Hence the total contribution of this run to
\(\sum_{x\in F}d(x,P)\) is at most $b\ell/2$.

Let \(M_{\mathrm{run}}\) be the maximum length of a blue run along \(T\).
Summing over all blue runs, and using that the corresponding red-to-red
subpaths are edge-disjoint, we obtain
\[
    \sum_{x\in F}d(x,P)
    \le
    \frac{M_{\mathrm{run}}}{2}\sum \ell
    \le
    \frac{M_{\mathrm{run}}}{2}L_T
    \le
    M_{\mathrm{run}}\OPT(X).
\]
It remains to bound \(M_{\mathrm{run}}\). For any fixed cyclic interval of
\(t\) consecutive vertices along \(T\), if \(t>L\), then the probability that
all \(t\) vertices are blue is zero. Otherwise, this probability is
\[
    \frac{\binom{M+L-t}{L-t}}{\binom{M+L}{L}}
    =
    \prod_{i=0}^{t-1}\frac{L-i}{M+L-i}
    \le
    2^{-t},
\]
where the last inequality uses \(L\le M\). There are at most
\(M+L\le 2n\) cyclic intervals of length \(t\). Therefore,
$
    \Pr[M_{\mathrm{run}}\ge t]\le 2n\cdot 2^{-t}.
$
Choosing
$
    t=\left\lceil (A+4)\log_2 n\right\rceil
$
gives
$
    \Pr[M_{\mathrm{run}}\ge t]\le 2n^{-(A+2)}.
$
On the complementary event,
\[
    \sum_{x\in F}d(x,\Gamma)
    \le
    \sum_{x\in F}d(x,P)
    \le
    M_{\mathrm{run}}\OPT(X)
    \le
    C_A\log n\cdot \OPT(X),
\]
for a suitable constant \(C_A\). This proves the lemma.
\end{proof}

We now bound the cost incurred in one phase. At the beginning of the phase,
the curve \(\Gamma\) is fixed, since it depends only on the prefix \(P\).
Conditional on \(P\) and on the set \(F\), the points of \(F\) arrive in
uniformly random order. Thus the projected keys form a one-dimensional
random-order online sorting instance of size \(L\).

The projected keys lie in \([0,\ell(\Gamma)]\). If \(\ell(\Gamma)>0\), we
apply the one-dimensional sorting algorithm to the rescaled keys
\(\theta_\Gamma(x)/\ell(\Gamma)\); scaling back, the projected cost is
multiplied by \(\ell(\Gamma)\). If \(\ell(\Gamma)=0\), all projected keys are
identical and the projected cost is zero.

Applying Theorem~\ref{thm:sorting-parameterized} with ambient parameter \(N=n\)
and failure exponent \(A+2\), the projected cost of the ordering produced in
the phase is at most $O_A(\log^2 n)\cdot \ell(\Gamma)$ with probability at least \(1-n^{-(A+2)}\).
 Let \(\sigma\) be this ordering.
By Lemma~\ref{lem:projection}, the cost of the same ordering in the original
metric is at most
\[
    O_A(\log^2 n)\cdot \ell(\Gamma)
    +
    2\sum_{x\in F}d(x,\Gamma).
\]
Combining Proposition~\ref{prop:prefix-curve-length} and
Lemma~\ref{lem:projection-distance}, we obtain that, with probability at
least \(1-O(n^{-(A+2)})\), the contribution of the phase is at most
\[
    O_A(\log^2 n)\cdot \OPT(X)
    +
    O_A(\log n)\cdot \OPT(X)
    =
    O_A(\log^2 n)\cdot \OPT(X).
\]

We conclude with the main theorem for the multidimensional problem.

\begin{theorem}\label{thm:random-order-tsp}
There exists an absolute constant \(C>0\) such that, for every adversarially
chosen multiset \(X\subseteq[0,1]^d\) of size \(n\),
\[
    \Pr\!\left[
        \cost(A)\le C\log^3 n\cdot \OPT(X)
    \right]
    \ge 1-n^{-4},
\]
where \(A\) is the array returned by Algorithm~\ref{alg:tsp} when the points of
\(X\) are revealed in uniformly random order.
\end{theorem}

\begin{proof}
Fix a sufficiently large constant \(A\). There are \(O(\log n)\) phases,
because the prefix size doubles in every non-final phase. By the phase bound
above and a union bound, with probability at least \(1-n^{-A}\), every phase
contributes at most
$
    O_A(\log^2 n)\cdot \OPT(X).
$
Summing over all phases gives total internal phase cost
$
    O_A(\log^3 n)\cdot \OPT(X).
$

The initial block and the final leftover block, if any, contain
\(O(\log^2 n)\) points in total and are placed arbitrarily. Since the distance
between any two input points is at most \(\OPT(X)\), their total contribution
is at most
$
    O(\log^2 n)\cdot \OPT(X),
$
which is dominated by the phase cost.

It remains to account for the boundaries between consecutive blocks in the
array. There are \(O(\log n)\) such boundaries. The cost of each boundary edge
is at most \(\OPT(X)\), since the distance between any two input points is at
most the length of an optimal Hamiltonian path through all points. Hence the
total boundary cost is
$
    O(\log n)\cdot \OPT(X),
$
which is again dominated by the phase cost. Therefore the total cost is
\[
    O_A(\log^3 n)\cdot \OPT(X)
\]
with probability at least \(1-n^{-A}\). Taking \(A=4\) proves Theorem~\ref{thm:random-order-tsp}.
\end{proof}

\section{Conclusion}

We initiated the study of Online Sorting in the random-order model, where the input multiset is chosen adversarially but revealed in uniformly random order. Our results show that the distributional information used by previous stochastic algorithms can instead be replaced by information extracted online from random prefixes of the input itself. We also extend the guarantees known for stochastic Online TSP beyond the uniform distribution, obtaining an \(O(\log^3 n)\)-competitive algorithm with high probability.

A natural next direction is to study more models that interpolate further between the i.i.d.\ stochastic setting and fully adversarial inputs. One particularly intriguing candidate is a \emph{prophet} version of Online Sorting. In this model, at each time step \(t\), the arriving element \(x_t\) is drawn independently from a known distribution \(D_t\), where the distributions may vary with time. The algorithm knows the sequence \(D_1,\ldots,D_n\), but observes each realized value only upon arrival and must place it irrevocably.

This model poses a qualitatively different challenge from both the stochastic and random-order settings studied so far. In the i.i.d.\ model, every arrival is governed by the same distribution, while in the random-order model, the random permutation provides a strong form of symmetry across time. The prophet model loses this time uniformity: early and late arrivals may follow entirely different distributions, so a prefix or time window need not be representative of the instance as a whole. It therefore becomes unclear whether the sampling-based ideas underlying the known stochastic and random-order algorithms can still be made to work. Determining whether polylogarithmic competitive ratios remain achievable in this setting is an intriguing open problem.

\newpage
\bibliographystyle{plainnat}
\bibliography{refs}

\appendix
\section{Random Order Online Sorting Deferred Proofs}

\subsection{Proof of Lemma~\ref{lem:gap-sum-concentration}} \label{App:ConcentrationProof}

\begin{proof}
If \(q=M+1\), then \(Q=\{0,\ldots,M\}\) and \[ G_Q=\sum_{r=0}^M g_r=L=\mu_Q \] deterministically, so the claim is immediate. Hence assume \(q\le M\).

By exchangeability, it suffices to consider $Q=\{0,1,\ldots,q-1\}$. Then \(G_Q\) is the number of zeros appearing before the \(q\)-th one in a uniformly random binary string containing \(M\) ones and \(L\) zeros.

For an integer \(k\), let \(Z_k\) denote the number of zeros among the first
\(k\) positions of this string. Then \(Z_k\) is hypergeometrically distributed
with mean
\[
    \lambda_k=\frac{kL}{M+L}.
\]
We use the standard Chernoff--Bernstein bounds for hypergeometric random
variables:
\[
    \Pr[Z_k\ge\lambda_k+t]
    \le
    \exp\left(
        -\frac{t^2}{2(\lambda_k+t/3)}
    \right)
\]
and
\[
    \Pr[Z_k\le\lambda_k-t]
    \le
    \exp\left(
        -\frac{t^2}{2\lambda_k}
    \right).
\]

We first consider the upper tail. Set $s=\lceil\mu_Q+a\rceil$. If \(s>L\), the event \(G_Q\ge s\) is impossible. Otherwise, setting $k=q+s-1$, we have $\{G_Q\ge s\}=\{Z_k\ge s\}$.

Writing \(d=s-\lambda_k\), we obtain
\[
\begin{aligned}
    d
    &=
    \frac{Ms-L(q-1)}{M+L}\\
    &\ge
    \frac{M(\mu_Q+a)-L(q-1)}{M+L}\\
    &=
    \frac{Ma+L-\mu_Q}{M+L}\\
    &\ge
    \frac{a}{1+C_0},
\end{aligned}
\]
where we used \(\mu_Q\le L\). Moreover, $\lambda_k+d=s\le\mu_Q+a+1$.

The hypergeometric upper-tail bound therefore gives
\[
    \Pr[G_Q\ge\mu_Q+a]
    \le
    \exp\left(
        -c\frac{a^2}{\mu_Q+a+1}
    \right)
    \le
    \exp\left(
        -c'\min\left\{
            \frac{a^2}{\mu_Q+1},
            a
        \right\}
    \right),
\]
for constants \(c,c'>0\) depending only on \(C_0\).

For the lower tail, we may assume \(a\le\mu_Q\), since otherwise the event is empty. Set $s=\lfloor\mu_Q-a\rfloor,$ and $k=q+s.$ Then $\{G_Q\le s\}=\{Z_k\le s\}$. Writing \(d=\lambda_k-s\), we have
\[
\begin{aligned}
    d
    &=
    \frac{Lq-Ms}{M+L}\\
    &\ge
    \frac{Lq-M(\mu_Q-a)}{M+L}\\
    &=
    \frac{\mu_Q+Ma}{M+L}\\
    &\ge
    \frac{a}{1+C_0}.
\end{aligned}
\]
Furthermore,
\[
    \lambda_k
    \le
    \frac{qL}{M+L}
    +
    \frac{sL}{M+L}
    \le
    2\mu_Q.
\]
Hence the hypergeometric lower-tail bound gives
\[
    \Pr[G_Q\le\mu_Q-a]
    \le
    \exp\left(
        -c\frac{a^2}{\mu_Q+1}
    \right).
\]
Combining the two tails proves the claim.
\end{proof}

\subsection{Proof of Lemma \ref{lem:regularity-implies-invariant}}
\label{app:regularity-implies-invariant}

\begin{proof}
We prove the following stronger quantitative statement. At every non-final update, immediately before the new phase is created, every bucket of the immediately preceding phase has $\Delta\pm\Gamma$ empty cells, while every bucket of the phase before it is full. Immediately
before the final arrival \(x_n\), all buckets created before
\(\Phi^\downarrow_K\) are full and the unique bucket of
\(\Phi^\downarrow_K\) has exactly one empty cell.

Since
\[
    \Delta=\Theta(\log^{3/2}n)
    \qquad\text{and}\qquad
    \Gamma\le\frac{\Delta}{100},
\]
this statement implies the main invariant after choosing
\(c_a,c_m,c_d\) sufficiently large.

We proceed inductively over the phase updates. Throughout the proof, the
notation \(\pm\Gamma\) absorbs both the deviations guaranteed by
\(\mathcal E_{\mathrm{reg}}\) and all additive \(O(1)\) errors arising from
the balanced empirical rank cutoffs, integer apportionment of bucket sizes,
and rounding of the time grid.

We first record two observations that will be used repeatedly.

First, the average effective capacity of a newly created phase is determined
by the time grid. Immediately before ascent phase
\(\Phi^\uparrow_j\) is created, the number of cells already allocated is
\(\tau_{K-j+1}\), while  $t^\uparrow_j=(1-\eta)\tau_{K-j+1}.$
Thus the total number of empty cells in previously allocated phases is $\eta\tau_{K-j+1}.$
Since \(\Phi^\uparrow_j\) has \(4^j\) intervals, the average charge per new
interval is
\[
    \frac{\eta\tau_{K-j+1}}{4^j}
    =
    \frac{\Delta}{4}.
\]
Since the fresh-cell budget contributes \(b\) cells per interval on average,
the balanced extension rule gives effective capacity $b+\Delta/4\pm O(1)$
for every ascent interval.

Similarly, immediately before the middle phase is created, the total number of
empty cells is \(\eta\tau_0\). Since the middle phase has \(4^K\) intervals,
the average charge per middle interval is
\[
    \frac{\eta\tau_0}{4^K}
    =
    \Delta.
\]
Hence every middle interval has effective capacity $b+\Delta\pm O(1)$.

Finally, immediately before descent phase \(\Phi^\downarrow_j\) is created,
the number of cells already allocated is \(n-\tau_j\), whereas $t^\downarrow_j
    =
    n-(1+4\eta)\tau_j$.
Thus the previously allocated phases contain $4\eta\tau_j$
empty cells. Since \(\Phi^\downarrow_j\) has \(4^{K-j}\) intervals, the
average charge per new interval is
\[
    \frac{4\eta\tau_j}{4^{K-j}}
    =
    4\Delta.
\]
Hence every non-final descent interval has effective capacity $b+4\Delta\pm O(1)$.

Second, suppose that, at some update, a group of residual cells from the
preceding phase is assigned eligibility interval \(J\). By the induction
hypothesis, this group contains at most $\Delta+\Gamma=o(\tau_K)$ cells. On \(\mathcal E_{\mathrm{reg}}\), the interval \(J\) receives
\(\Omega(\tau_K)\) arrivals during the relevant one-step window. Moreover,
all still earlier phases are full by induction, and the algorithm always
places an arriving element into the earliest created bucket containing an
eligible empty cell. Consequently, all residual cells eligible for \(J\) are
filled before any later-created ordinary cells eligible for \(J\) are used.

During the ascent and the middle phase, the required load estimate follows
from \textnormal{(R1)}. During the descent it follows from
\textnormal{(R2)}, and in the final period from \textnormal{(R3)}. Once these
residual cells have been filled, the number of empty cells in the new bucket
can be computed by subtracting the total load of its interval from its total
effective capacity.

We now verify the induction phase by phase.

\paragraph{Initial phase.}

The initial phase consists of one bucket of size \(\tau_K\). The first ascent
update occurs at time $t^\uparrow_1=(1-\eta)\tau_K.$
Hence, immediately before \(\Phi^\uparrow_1\) is created, the initial bucket
has exactly $\tau_K-(1-\eta)\tau_K
    =
    \Delta$
empty cells. This establishes the base case.

\paragraph{Ascent phases.}

Fix \(j\in\{1,\ldots,K\}\), and suppose inductively that immediately before
\(\Phi^\uparrow_j\) is created, every bucket of
\(\Phi^\uparrow_{j-1}\) has $\Delta\pm\Gamma$
empty cells, while every bucket of \(\Phi^\uparrow_{j-2}\) is full whenever
\(j\ge2\).

The established bound on the empty cells of
\(\Phi^\uparrow_{j-1}\) is precisely the property used in the proof of
\cref{lem:ascent-containment}. The same argument therefore shows that every
interval of \(\Phi^\uparrow_{j-1}\) contains an interval of
\(\Phi^\uparrow_j\), so the charging rule is well-defined.

Every new interval receives charge from at most one old interval. Hence $c(J)\le\Delta+\Gamma$.
Since every new interval has effective capacity $ b+\Delta/4\pm O(1)$,
its fresh bucket has size at least
\[
    b+\frac{\Delta}{4}
    -(\Delta+\Gamma)
    -O(1)
    =
    b-\frac{3\Delta}{4}-\Gamma-O(1)
    >0.
\]
Thus the balanced extension is feasible.

We now consider the one-step window following the creation of
\(\Phi^\uparrow_j\). If \(j<K\), then
\[
    t^\uparrow_{j+1}-t^\uparrow_j
    =
    (1-\eta)
    \bigl(
        \tau_{K-j}-\tau_{K-j+1}
    \bigr).
\]
Since \(\Phi^\uparrow_j\) has \(4^j\) balanced empirical intervals, the
expected one-step load of every interval is $(1-\eta)b\pm O(1)$.
For \(j=K\), the next update is the middle update, and the same calculation
gives expected one-step load $(1-\eta)b\pm O(1)$.
Therefore, by \textnormal{(R1)}, every interval of
\(\Phi^\uparrow_j\) receives $(1-\eta)b\pm\Gamma$
arrivals during this window.

Every residual group created at this update is assigned to some interval
\(J\in\mathcal I(\Phi^\uparrow_j)\). Such a group contains at most $\Delta+\Gamma=o(\tau_K)$
cells, whereas the interval \(J\) receives $(1-\eta)b\pm\Gamma
    =
    \Omega(\tau_K)$
arrivals before the next update. By the earliest-created-bucket placement
rule, every such residual group is therefore filled during this window.
Consequently, every bucket of \(\Phi^\uparrow_{j-1}\) is full by the next
update.

We may therefore compute the empty cells remaining in each new ascent bucket
by subtracting its interval load from its effective capacity. Using $\eta b=3\Delta/4$,
we obtain
\[
\begin{aligned}
    b+\frac{\Delta}{4}
    -(1-\eta)b
    \pm\Gamma
    &=
    \Delta\pm\Gamma.
\end{aligned}
\]
Thus, at the next update, every bucket of \(\Phi^\uparrow_j\) has $\Delta\pm\Gamma$
empty cells, while every bucket of \(\Phi^\uparrow_{j-1}\) is full. This
completes the ascent induction.

\paragraph{Middle phase.}

Immediately before the middle phase is created, every bucket of
\(\Phi^\uparrow_K\) has $\Delta\pm\Gamma$
empty cells, while every bucket of \(\Phi^\uparrow_{K-1}\) is full.

The charging map is the identity, so each middle interval receives at most $\Delta+\Gamma$ units of charge. Since its effective capacity is $b+\Delta\pm O(1)$,
every fresh middle bucket has size at least $b-\Gamma-O(1)>0$.
Thus the balanced extension is feasible.

Using $t^\downarrow_1
    =
    n-(1+4\eta)\tau_1$
and $    t^{\mathrm{mid}}
    =
    (1-\eta)\tau_0,$
together with
$
    n=2\tau_0
    \text{ and }
    \tau_0=4\tau_1,
$
we obtain
\[
    t^\downarrow_1-t^{\mathrm{mid}}
    =
    \tau_0-\tau_1.
\]
Since the middle phase has \(4^K\) intervals, the expected one-step load of
every middle interval is $b\pm O(1)$.
Hence, by \textnormal{(R1)}, every middle interval receives $b\pm\Gamma$
arrivals before \(\Phi^\downarrow_1\) is created.

The residual cells inherited from \(\Phi^\uparrow_K\) retain the same
eligibility intervals as the corresponding middle intervals. Each such
residual group contains at most $\Delta+\Gamma=o(\tau_K)$
cells, while its eligibility interval receives
$ b\pm\Gamma
    =
    \Omega(\tau_K)$
arrivals during this window. Hence every residual group is filled, and
therefore every bucket of \(\Phi^\uparrow_K\) is full before
\(\Phi^\downarrow_1\) is created.

Subtracting the one-step load from the effective capacity of each middle
interval gives $b+\Delta-b\pm\Gamma
    =
    \Delta\pm\Gamma$.
Thus, immediately before \(\Phi^\downarrow_1\) is created, every middle
bucket has $\Delta\pm\Gamma$
empty cells, while every bucket of \(\Phi^\uparrow_K\) is full.

\paragraph{Descent phases before the final one.}

Fix \(j<K\), and suppose inductively that immediately before
\(\Phi^\downarrow_j\) is created, every bucket of the preceding phase has $\Delta\pm\Gamma$
empty cells, while every bucket of the phase before that is full.

Each interval of \(\Phi^\downarrow_j\) is the union of four intervals of the
preceding phase. Hence its total charge is at most $4\Delta+4\Gamma.$
Since its effective capacity is $b+4\Delta\pm O(1),$
every fresh bucket has size at least $b-4\Gamma-O(1)>0.$
Thus the balanced extension is feasible.

The four residual groups charged to a merged interval retain their original
child eligibility intervals. By \textnormal{(R2)}, every such child interval
receives $\Omega(\tau_K)$
arrivals during the following one-step window, whereas the corresponding
residual group contains at most $\Delta+\Gamma=o(\tau_K)$
cells. Since the algorithm prioritizes earlier-created eligible cells, every
residual group is filled during this window. Consequently, every bucket of the
phase preceding \(\Phi^\downarrow_j\) is full by the next update.

The next descent window has length
\[
    t^\downarrow_{j+1}-t^\downarrow_j
    =
    (1+4\eta)(\tau_j-\tau_{j+1}).
\]
Since \(\Phi^\downarrow_j\) has \(4^{K-j}\) intervals, the expected one-step
load of every interval is $(1+4\eta)b\pm O(1)$.
By \textnormal{(R1)}, the actual load is therefore $(1+4\eta)b\pm\Gamma.$
Using $4\eta b=3\Delta$,
the number of empty cells remaining in each new descent bucket at the next
update is $\begin{aligned}
    b+4\Delta
    -(1+4\eta)b
    \pm\Gamma
    &=
    \Delta\pm\Gamma.
\end{aligned}$

Thus, at the next update, every bucket of \(\Phi^\downarrow_j\) has $\Delta\pm\Gamma$
empty cells, while every bucket of the preceding phase is full. This completes
the induction over all non-final descent phases.

\paragraph{Final descent phase.}

Immediately before \(\Phi^\downarrow_K\) is created, every bucket of
\(\Phi^\downarrow_{K-1}\) has $\Delta\pm\Gamma$
empty cells, while every bucket of \(\Phi^\downarrow_{K-2}\) is full. The
algorithm has already placed
\[
    t^\downarrow_K
    =
    n-(1+4\eta)\tau_K
\]
elements. At this point, all fresh cells except the \(\tau_K\) cells reserved
for \(\Phi^\downarrow_K\) have already been allocated. Hence the number of
empty cells in previously allocated buckets is exactly
\[
\begin{aligned}
    n-t^\downarrow_K-\tau_K
    &=
    (1+4\eta)\tau_K-\tau_K
    =
    4\Delta.
\end{aligned}
\]
Since all phases before \(\Phi^\downarrow_{K-1}\) are full, these are exactly
the residual cells inherited from the four buckets of
\(\Phi^\downarrow_{K-1}\).

All \(4\Delta\) cells are charged to the unique interval \([0,1]\) of the
final descent phase while retaining their original child eligibility
intervals. The unique fresh bucket of \(\Phi^\downarrow_K\) has size
\(\tau_K\). Therefore the total effective capacity after the final descent
update is $\tau_K+4\Delta
    =
    (1+4\eta)\tau_K,$
which is exactly the number of arrivals remaining after time
\(t^\downarrow_K\).

There are four residual groups, corresponding to the four intervals of
\(\Phi^\downarrow_{K-1}\). By \textnormal{(R3)}, every such child interval
receives $\Omega(\tau_K)$
arrivals before time \(n\), while its residual group contains at most $\Delta+\Gamma=o(\tau_K)$
cells. Hence every residual group is filled before the final arrival, and
therefore all buckets created before \(\Phi^\downarrow_K\) are full before
\(x_n\) arrives.

Immediately before \(x_n\) is placed, exactly one cell of the entire array is
empty. Since all earlier buckets are full, this cell belongs to the unique
bucket of \(\Phi^\downarrow_K\). Thus this bucket has exactly one empty cell.

The strengthened quantitative statement follows, and therefore $\mathcal E_{\mathrm{reg}}
    \subseteq
    \mathcal G.$
\end{proof}

\end{document}